\documentclass{eptcs}
\usepackage{bm}
\usepackage{tikzit}
\usepackage{amssymb, amsthm, amsmath}
\theoremstyle{definition}
\newtheorem{definition}{Definition}[section]
\newtheorem{lemma}[definition]{Lemma}
\newtheorem{theorem}[definition]{Theorem}
\newtheorem{proposition}[definition]{Proposition}
\newtheorem{corollary}[definition]{Corollary}
\newtheorem{example}[definition]{Example}
\def\<{\langle}
\def\>{\rangle}
\newcommand{\Odd}{\mathop{Odd}}

\input{quantum.tikzdefs}
\tikzstyle{gate}=[shape=rectangle, text height=1.5ex, text depth=0.25ex, yshift=-0.5mm, fill=white, draw=black, minimum height=5mm, minimum width=5mm, font={\small}, tikzit category=circuit]
\tikzstyle{big gate}=[shape=rectangle, text height=1.5ex, text depth=0.25ex, yshift=-0.5mm, fill=white, draw=black, minimum height=10mm, minimum width=5mm, font={\small}, tikzit category=circuit]
\tikzstyle{Z dot}=[inner sep=0mm, minimum size=2mm, shape=circle, draw=black, fill=zxgreen, tikzit fill={rgb,255: red,221; green,255; blue,221}, tikzit category=zx]
\tikzstyle{Z bold dot}=[inner sep=0mm, minimum size=2mm, shape=circle, draw=black, fill=zxgreen, tikzit fill={rgb,255: red,221; green,255; blue,221}, line width=1.2pt, tikzit category=zx]
\tikzstyle{Z phase dot}=[minimum size=5mm, font={\footnotesize\boldmath}, shape=rectangle, rounded corners=2mm, inner sep=0.2mm, outer sep=-2mm, scale=0.8, tikzit shape=circle, draw=black, fill=zxgreen, tikzit fill={rgb,255: red,221; green,255; blue,221}, tikzit draw=blue, tikzit category=zx]
\tikzstyle{Z tiny dot}=[inner sep=0.2mm, font={\footnotesize\boldmath}, minimum size=1mm, shape=circle, draw=black, fill=zxgreen, tikzit fill={rgb,255: red,221; green,255; blue,221}]
\tikzstyle{X dot}=[Z dot, shape=circle, draw=black, fill=zxred, tikzit fill={rgb,255: red,255; green,136; blue,136}, tikzit category=zx]
\tikzstyle{X bold dot}=[inner sep=0mm, minimum size=2mm, shape=circle, draw=black, fill=zxred, tikzit fill={rgb,255: red,255; green,136; blue,136}, line width=1.2pt, tikzit category=zx]
\tikzstyle{X phase dot}=[Z phase dot, tikzit shape=circle, tikzit draw=blue, fill=zxred, tikzit fill={rgb,255: red,255; green,136; blue,136}, font={\footnotesize\boldmath}, tikzit category=zx]
\tikzstyle{X tiny dot}=[inner sep=0.2mm, font={\footnotesize\boldmath}, minimum size=1mm, shape=circle, draw=black, fill=zxred, tikzit fill={rgb,255: red,255; green,136; blue,136}]
\tikzstyle{hadamard}=[fill=yellow, draw=black, shape=rectangle, inner sep=0.6mm, minimum height=1.5mm, minimum width=1.5mm, tikzit category=zx]
\tikzstyle{paulibox}=[fill={rgb,255: red,221; green,221; blue,255}, draw=black, shape=rectangle, inner sep=0.6mm, minimum height=5mm, minimum width=5mm, font={\footnotesize}, text height=1.5ex, text depth=0.25ex, tikzit category=zx]
\tikzstyle{vertex}=[inner sep=0.2mm, minimum size=1mm, shape=circle, draw=black, fill=black, tikzit category=misc]
\tikzstyle{vertex set}=[inner sep=0.2mm, minimum size=1mm, shape=circle, draw=black, fill=white, font={\footnotesize\boldmath}, tikzit category=misc]
\tikzstyle{small black dot}=[fill=black, draw=black, shape=circle, inner sep=0pt, minimum width=1.2mm, tikzit category=circuit]
\tikzstyle{cnot ctrl}=[fill=black, draw=black, shape=circle, inner sep=0pt, minimum width=1.2mm, tikzit category=circuit]
\tikzstyle{cnot targ}=[fill=white, draw=white, shape=circle, tikzit category=circuit, label={center:$\oplus$}, inner sep=0pt, minimum width=2.1mm, tikzit fill={rgb,255: red,102; green,204; blue,255}, tikzit draw=black]
\tikzstyle{ket}=[fill=white, draw=black, shape=regular polygon, regular polygon sides=3, regular polygon rotate=-30, scale=0.7, inner sep=1pt, tikzit category=circuit, tikzit shape=rectangle, tikzit fill=green]
\tikzstyle{bra}=[fill=white, draw=black, shape=regular polygon, regular polygon sides=3, regular polygon rotate=30, scale=0.7, inner sep=1pt, tikzit category=circuit, tikzit shape=rectangle, tikzit fill=red]
\tikzstyle{scalar}=[shape=rectangle, text height=1.5ex, text depth=0.25ex, yshift=-0.5mm, fill=white, draw=black, minimum height=5mm, minimum width=5mm, font={\small}]
\tikzstyle{clabel}=[fill=white, draw=none, shape=rectangle, tikzit fill={rgb,255: red,56; green,255; blue,242}, font={\footnotesize}, inner sep=1pt, tikzit category=labels]
\tikzstyle{empty diagram}=[draw=gray!40!white, dashed, shape=rectangle, minimum width=1cm, minimum height=1cm, tikzit category=misc]
\tikzstyle{amap}=[fill=white, draw=black, shape=NEbox, tikzit category=asymmetric, tikzit fill=yellow, tikzit shape=rectangle]
\tikzstyle{amap conj}=[fill=white, draw=black, shape=NWbox, tikzit category=asymmetric, tikzit fill=green, tikzit shape=rectangle]
\tikzstyle{amap adj}=[fill=white, draw=black, shape=SEbox, tikzit category=asymmetric, tikzit fill=red, tikzit shape=rectangle]
\tikzstyle{amap trans}=[fill=white, draw=black, shape=SWbox, tikzit category=asymmetric, tikzit fill=orange, tikzit shape=rectangle]
\tikzstyle{astate}=[fill=white, draw=black, shape=NEtriangle, tikzit category=asymmetric, tikzit shape=circle, tikzit fill=yellow]
\tikzstyle{astate conj}=[fill=white, draw=black, shape=NWtriangle, tikzit category=asymmetric, tikzit shape=circle, tikzit fill=green]
\tikzstyle{astate adj}=[fill=white, draw=black, shape=SEtriangle, tikzit category=asymmetric, tikzit shape=circle, tikzit fill=red]
\tikzstyle{astate trans}=[fill=white, draw=black, shape=SWtriangle, tikzit category=asymmetric, tikzit shape=circle, tikzit fill=orange]
\tikzstyle{box}=[shape=rectangle, text height=1.5ex, text depth=0.25ex, yshift=-0.5mm, fill=white, draw=black, minimum height=5mm, minimum width=5mm, font={\small}]
\tikzstyle{medium box}=[shape=rectangle, text height=1.5ex, text depth=0.25ex, yshift=-0.5mm, fill=white, draw=black, minimum height=10mm, minimum width=5mm, font={\small}]
\tikzstyle{simple}=[-]
\tikzstyle{hadamard edge}=[-, dashed, dash pattern=on 2pt off 0.5pt, thick, draw={rgb,255: red,68; green,136; blue,255}]
\tikzstyle{box edge}=[-, dashed, dash pattern=on 2pt off 0.5pt, thick, draw={rgb,255: red,203; green,192; blue,225}]
\tikzstyle{brace edge}=[-, tikzit draw=blue, decorate, decoration={brace,amplitude=1mm,raise=-1mm}]
\tikzstyle{diredge}=[->]
\tikzstyle{double edge}=[-, double, shorten <=-1mm, shorten >=-1mm, double distance=2pt]
\tikzstyle{gray edge}=[-, gray!60!white]
\tikzstyle{pointer edge}=[->, very thick, gray]
\tikzstyle{boldedge}=[-, line width=1.2pt, shorten <=-0.17mm, shorten >=-0.17mm]
\tikzstyle{bidir edge}=[<->, very thick, draw={rgb,255: red,191; green,191; blue,191}]
\tikzstyle{surface X}=[-, tikzit fill=red, fill=zxred]
\tikzstyle{surface Z}=[-, tikzit fill=green, fill=zxgreen]
\tikzstyle{X Web}=[-, preaction={line width=1.5mm, draw={rgb,255: red,255; green,150; blue,150}, dashed, dash pattern=on 2.25pt off 0.5pt}, shorten <=-0.25mm, shorten >=-0.25mm, tikzit draw=red]
\tikzstyle{Y Web}=[-, preaction={line width=1.5mm, draw={rgb,255: red,150; green,150; blue,255}, dashed, dash pattern=on 1.2pt off 0.5pt}, shorten <=-0.25mm, shorten >=-0.25mm, tikzit draw=blue]
\tikzstyle{Z Web}=[-, preaction={line width=1.5mm, draw={rgb,255: red,120; green,200; blue,120}}, shorten <=-0.25mm, shorten >=-0.25mm, tikzit draw={rgb,255: red,0; green,141; blue,0}]
\tikzstyle{light-fill}=[-, fill={rgb,255: red,230; green,230; blue,230}, tikzit fill={rgb,255: red,191; green,191; blue,191}, draw=none, tikzit draw={rgb,255: red,191; green,191; blue,191}]
\tikzstyle{dashed edge}=[-, dashed, dash pattern=on 2pt off 0.5pt, draw=black]
\tikzstyle{light dashed edge}=[-, dashed, dash pattern=on 2pt off 0.5pt, draw={rgb,255: red,191; green,191; blue,191}]

\title{ZX-Flow: A Flexible Criterion for Deterministic \\ Computation with ZX-Diagrams}

\author{Aleks Kissinger
\institute{University of Oxford}
\email{aleks.kissinger@ox.ac.uk} \and
John van de Wetering
\institute{University of Amsterdam, QuSoft}
\email{john@vdwetering.name}
}

\def\titlerunning{ZX-Flow}
\def\authorrunning{Aleks Kissinger and John van de Wetering}

\begin{document}

\maketitle

\begin{abstract}
Flow criteria are used to efficiently extract computations, either in the form of measurement patterns or quantum circuits, from ZX-diagrams. Existing criteria such as causal flow, generalised flow, and Pauli flow, were all originally formulated for graph states, so they require ZX-diagrams to be in a very particular graph-state-like form. This form is easily broken by applying basic ZX rules and makes establishing some desirable properties very complicated. Here, we introduce a new ``ZX-native'' flow criterion called ZX-flow, formulated using a new type of decoration of a ZX-diagram we call Pauli semiwebs. These are a generalisation of Pauli webs, which have recently been used extensively in reasoning about fault-tolerant computations in the ZX-calculus. We show that ZX-flow is straightforwardly preserved by all Clifford rewrites and furthermore that a ZX-diagram has ZX-flow if and only if it is Clifford-equivalent to a graph-like ZX-diagram with Pauli flow. Finally, we show that any diagram with ZX-flow can be readily interpreted either as a deterministic measurement-based computation or as a Clifford isometry followed by a sequence of Pauli exponentials. The latter can then be efficiently extracted to a quantum circuit.
\end{abstract}

\section{Introduction}


The ZX-calculus~\cite{CD1} is a graphical language for reasoning about quantum processes that has many applications in quantum computation. It allows one to represent a computation as a graph-like structure called a ZX-diagram and rewrite it according to a complete set of graphical rules~\cite{Backens1,vilmarteulercompleteness} to transform and simplify quantum computations. It has been used for instance to optimise and verify quantum circuits~\cite{duncan2019graph,Vandaele2025thesis,fischbach2025review,kissinger2019tcount,thanos2024automated,hurwitz2024simulation}, analyse parametrised quantum circuits~\cite{peham2022equivalence,Wetering2025optimalcompilation,stollenwerk2024measurementbased,stollenwerk2023diagrammatic,wang2022differentiating}, develop new constructions in measurement-based quantum computing~\cite{Backens2020extraction,kissinger2017MBQC,defelice2026dataflow,ostmann2025nonlinear,EPTCS426.4}, and design fault-tolerant quantum computations~\cite{zhou2026topols,poor2025ultrlow,rodatz2025fault,h.s.derks2025dynamical,bauer2025planar}.

These methods benefit from the ZX-calculus because of its increased flexibility with respect to the standard quantum circuit notation, which allows representations of non-unitary operations including projections, ancilla preparations, measurements and post-selections. This additional flexibility does however come with a cost: when you produce a ZX-diagram representing a computation, it is usually not immediately clear how to translate this diagram into a representation that allows direct execution on quantum hardware. We have to first \emph{extract} a circuit from the diagram. This circuit extraction problem is very hard in the most general case~\cite{debeaudrap2022circuit}, but there are several known conditions on diagrams that allow efficient extraction of a circuit. These so-called \emph{flow} conditions were borrowed from the field of measurement-based quantum computation (MBQC), where they were originally used to associate single-qubit measurements to  certain ``correction sets'', which identify future measurements that can be adapted to obtain a deterministic result. The most well-known such conditions are called generalised flow and Pauli flow~\cite{browne2007gflow}.

It turns out that these conditions also suffice to turn a ZX-diagram into a quantum circuit expressed in a fixed basis of 1- and 2-qubit gates~\cite{duncan2019graph,Backens2020extraction,Simmons2021Measurement}. 
When optimising quantum circuits with the ZX-calculus, it is useful to preserve flow conditions to ensure that we can extract a circuit from the final, reduced ZX-diagram. However, even some of the simplest ZX rules, like spider fusion, break flow conditions, making it awkward to transform ZX-diagrams in a flow-preserving way.

Previously, this has required one to express all ZX rewrites in terms of a handful of rules like local complementation and pivoting, which are then (arduously) proven to preserve generalised flow or Pauli flow~\cite{duncan2019graph,Backens2020extraction,Simmons2021Measurement,mcelvanney2023flowpreserving,EPTCS426.4}. There are two reasons this has been so complicated. The first is that flow conditions were originally designed for graph states, not ZX-diagrams, so it requires us to maintain ZX-diagrams in a very special, ``graph-state-like'' form for the conditions to even be formulated. The second reason is that these conditions require one to explicitly track the time ordering for the classically tractable, Clifford parts of a ZX-diagram, even though this ordering is to a large extent arbitrary. Pauli flow handles this to some extent, but it still imposes some conditions on Clifford nodes which block many ZX rewrites.


In this paper we introduce a new flow condition specifically designed to work on ZX-diagrams that we call \emph{ZX-flow}. It is built around the concept of \emph{Pauli semiwebs}, a new non-Clifford generalisation of Pauli webs. Pauli webs are a way of colouring the wires of a ZX-diagram to track the propagation of Pauli operators through it. They have been a useful tool in using ZX to reason about fault-tolerant computations, as they give an elegant way to formalise logical operators, stabilisers, and detecting sets of measurements associated with an error-correcting procedure~\cite{bombin2023unifying,rodatz2025fault}. However, Pauli webs must satisfy strong local consistency conditions which break at non-Clifford locations, i.e.~nodes where the phase parameter is not an integer multiple of $\frac\pi2$. Pauli semiwebs, on the other hand, allow certain limited violations of these conditions at certain locations, which we call \emph{defects}, that allow them to handle non-Clifford structure. This new structure enables us to give a compact, easily-visualised flow criterion. We say a diagram has ZX-flow if its non-Clifford nodes can be given a time ordering, and each node can be associated to a Pauli semiweb which only has defects at the node itself and non-Clifford nodes in its future. We can then relate this to existing flow criteria with our first main result, which states that a diagram has ZX-flow if and only if it is Clifford-equivalent to a graph-like diagram with Pauli flow.

This structure gives us (at least) two ways to directly interpret a ZX-diagram as a quantum computation. The first is using MBQC, where we can interpret non-Clifford nodes as measurements of a stabiliser resource state and defects as locations where we feed-forward measurement outcomes. For our second main result, we show that a diagram with ZX-flow can be interpreted as a unitary quantum circuit, where the Pauli semiwebs enable us to read off a particular circuit decomposition just using the time-ordering and the support of semiwebs on the outputs. The latter is similar in spirit to the circuit extraction for Pauli flow due to Simmons~\cite{Simmons2021Measurement}, but it is much more direct and the proof of correctness follows straightforwardly from the properties of Pauli semiwebs.

In addition to the notion of ZX-flow and our two main theorems, a third contribution of this paper is the basic theory of Pauli semiwebs, which can be seen as generalising Pauli webs with a notion of ``feed-forward''. This unifies notions from fault-tolerance and MBQC, and may have broader applications, e.g. in the design of fault-tolerant computations involving non-Clifford components.

\section{Preliminaries}\label{sec:prelims}
\subsection{ZX-diagrams}\label{sec:diagrams}

ZX-diagrams are string diagrams built from the following basic components:
\[
\big\llbracket\;\begin{tikzpicture}
	\begin{pgfonlayer}{nodelayer}
		\node [style=Z phase dot] (8) at (0, 0) {$\alpha$};
		\node [style=none] (9) at (-1, 0.75) {};
		\node [style=none] (10) at (-1, -0.75) {};
		\node [style=none] (11) at (1, 0.75) {};
		\node [style=none] (12) at (1, -0.75) {};
		\node [style=none, rotate=90] (13) at (-1, 0) {...};
		\node [style=none, rotate=90] (14) at (1, 0) {...};
		\node [style=none] (15) at (0, -1.25) {$Z$-spider};
	\end{pgfonlayer}
	\begin{pgfonlayer}{edgelayer}
		\draw (9.center) to (8);
		\draw (8) to (11.center);
		\draw (10.center) to (8);
		\draw (8) to (12.center);
	\end{pgfonlayer}
\end{tikzpicture} \big\rrbracket \ := \ 
|0\>^{\otimes m}\<0|^{\otimes n} + e^{i \alpha} |1\>^{\otimes m}\<1|^{\otimes n}
\qquad
\big\llbracket\;\begin{tikzpicture}
	\begin{pgfonlayer}{nodelayer}
		\node [style=X phase dot] (8) at (0, 0) {$\alpha$};
		\node [style=none] (9) at (-1, 0.75) {};
		\node [style=none] (10) at (-1, -0.75) {};
		\node [style=none] (11) at (1, 0.75) {};
		\node [style=none] (12) at (1, -0.75) {};
		\node [style=none, rotate=90] (13) at (-1, 0) {...};
		\node [style=none, rotate=90] (14) at (1, 0) {...};
		\node [style=none] (15) at (0, -1.25) {$X$-spider};
	\end{pgfonlayer}
	\begin{pgfonlayer}{edgelayer}
		\draw (9.center) to (8);
		\draw (8) to (11.center);
		\draw (10.center) to (8);
		\draw (8) to (12.center);
	\end{pgfonlayer}
\end{tikzpicture} \big\rrbracket \ := \ 
|{+}\>^{\otimes m}\<{+}|^{\otimes n} + e^{i \alpha} |{-}\>^{\otimes m}\<{-}|^{\otimes n}
\]
\[
\big\llbracket\;\begin{tikzpicture}
	\begin{pgfonlayer}{nodelayer}
		\node [style=hadamard] (16) at (0, 0) {};
		\node [style=none] (17) at (-1, 0) {};
		\node [style=none] (18) at (1, 0) {};
		\node [style=none] (19) at (0, -1.25) {$H$-gate};
	\end{pgfonlayer}
	\begin{pgfonlayer}{edgelayer}
		\draw (17.center) to (18.center);
	\end{pgfonlayer}
\end{tikzpicture}
 \big\rrbracket \ := \ 
|{+}\>\<0| + |{-}\>\<1|
\qquad
\big\llbracket\;\begin{tikzpicture}
	\begin{pgfonlayer}{nodelayer}
		\node [style=none] (26) at (-0.75, 0) {};
		\node [style=none] (27) at (0.75, 0) {};
		\node [style=none] (28) at (0, -1.15) {id};
	\end{pgfonlayer}
	\begin{pgfonlayer}{edgelayer}
		\draw (26.center) to (27.center);
	\end{pgfonlayer}
\end{tikzpicture}
 \big\rrbracket \ := \ 
\sum_i|i\>\<i|
\qquad
\big\llbracket\;\begin{tikzpicture}
	\begin{pgfonlayer}{nodelayer}
		\node [style=none] (29) at (-0.75, 0.5) {};
		\node [style=none] (30) at (-0.75, -0.5) {};
		\node [style=none] (31) at (0.75, 0.5) {};
		\node [style=none] (32) at (0.75, -0.5) {};
		\node [style=none] (33) at (0, -1.25) {swap};
	\end{pgfonlayer}
	\begin{pgfonlayer}{edgelayer}
		\draw (29.center) to (32.center);
		\draw (30.center) to (31.center);
	\end{pgfonlayer}
\end{tikzpicture} \big\rrbracket \ := \ 
\sum_{ij}|ij\>\<ji|
\]
using parallel and sequential composition. Two ZX-diagrams are considered equivalent if one can be deformed into the other without changing the order of inputs/outputs.

It will be convenient to represent ZX-diagrams combinatorically. We treat a diagram $D$ as a collection of nodes $N_D$ and wires $W_D$, which are related by a binary adjacency relation $a_D \subseteq N_D \times W_D$. Each node $\nu$ has a \textit{type} $t_D(\nu) \in \{ Z, X, H \}$ and a \textit{phase} $\phi_D(\nu) \in \mathbb R$. Wires can be adjacent to zero, one, or two nodes. We define totally ordered, not necessarily disjoint subsets $I_D, O_D \subseteq W_D$ of \textit{inputs} and \textit{outputs}. We also allow diagrams to include a scalar factor $\lambda_D \in \mathbb C$ and impose a meta-rule that there is a unique zero diagram. That is, for all diagrams $D, D'$ with the same number of inputs and outputs, $\lambda_D = \lambda_{D'} = 0 \implies D = D'$.

If $t_D(\nu) = H$, it represents an $H$-gate, must have phase $0$, and be adjacent to precisely two wires. If $t_D(\nu) \in \{ Z, X \}$, it is a \textit{spider}. 
Notably, the combinatoric structure of $D$ is enough to define the linear map $\llbracket D \rrbracket$ unambiguously. That is, the linear map represented by a ZX-diagram is fully specified by its nodes and connecting wires and isomorphic diagrams will define the same linear map. This principle is often called \textit{only connectivity matters} (see e.g.~\cite{KissingerWetering2024Book}). Note we drop the subscript $D$ from parts of the diagram and the semantic brackets $\llbracket . \rrbracket$ if they are clear from context.

It is sometimes useful to treat pairs of wires that share an $H$-gate as a single logical unit, especially when discussing the adjacency of spiders in a diagram.

\begin{definition}\label{def:edge}
  A pair of wires $\{ w_1, w_2 \}$ that are both adjacent to the same $H$-gate are called an \emph{$H$-edge}. Any single wire $\{ w \}$ that is not part of an $H$-edge is called a \emph{plain edge}. We consider two spiders to be adjacent if they are connected via a plain edge or an $H$-edge. Let $E_D \subseteq 2^{W}$ be the set of all plain edges and $H$-edges in $D$.
\end{definition}

The linear map of a ZX-diagram $D$ is also preserved by transformation according to a set of rewrite rules called the ZX-calculus. While there exist several variations, for our purposes we will focus on the \textit{extended Clifford ZX-calculus}, as shown in Figure~\ref{fig:Clifford-rules} of Section~\ref{sec:flow-preservation}.
We say a spider is \emph{Clifford} if its phase is a multiple of $\frac\pi2$, and is \emph{non-Clifford} otherwise. Diagrams consisting just of $H$-gates and Clifford spiders can be efficiently reduced to normal form~\cite{KissingerWetering2024Book} using the rules of the ZX-calculus. In particular, this implies efficient equality checking for such diagrams and a graphical version of the Gottesman-Knill theorem.

\subsection{Pauli webs}\label{sec:Pauli-webs}

A Pauli web is an annotation on a ZX-diagram that denotes how Pauli operators propagate through the diagram.

\begin{definition}
  Let $D$ be a ZX-diagram and let $\nu \in N_D$ be a node. The associated \emph{local state} $|\nu\>$ is defined as:
  $|0\>^{\otimes k} + e^{i \phi(\nu)}|1\>^{\otimes k}$ for $Z$-spiders,
  $|{+}\>^{\otimes k} + e^{i \phi(\nu)}|{-}\>^{\otimes k}$ for $X$-spiders, and
  $|0{+}\> + |1{-}\>$ for $H$-gates.
\end{definition}

Let $\mathcal{P}_n$ denote the Pauli group on $n$ qubits.
For a Pauli operator $\vec P \in \mathcal P_{|W|}$, let $\vec P_j$ and $\vec P_J$ be the Pauli operator restricted to a particular wire $j \in W$ and set of wires $J \subseteq W$, respectively. For a node $\nu$, we let $\vec P_{\nu} := \vec P_{a(\nu)}$ be $\vec P$ restricted to the neighborhood of $\nu$. As these restrictions are technically only defined up to a global phase, we fix the phase such that $\vec P_j$ and $\vec P_{\nu}$ are (tensor products of) the standard Pauli operators $\{I, X, Y, Z\}$.

\begin{definition}\label{def:pauli-web}
A \textit{Pauli web} for a diagram $D$ is a Pauli operator $\vec w \in \mathcal P_{|W|}$ such that $|\nu\>$ is an eigenstate for $\vec w_{\nu}$ for all $\nu \in N$. We say a Pauli web has \emph{$X$-support} at a wire $j$ if $\vec w_j$ is either $X$ or $Y$ and has \emph{$Z$-support} if it is either $Y$ or $Z$.
\end{definition}

We distinguish the four different ways a wire can be labelled by a Pauli web as follows:
\ctikzfig{web-labels}

Pauli webs track the propagation of Pauli operators through ZX-diagrams. In particular, they characterise the Pauli operators that leave a Clifford map invariant.

\begin{theorem}\label{thm:pauli-web-stabilisers}
Let $\vec w$ be a Pauli web for a ZX-diagram $D$, then we have $D = (-1)^k \vec w_O D \vec w_I$, where $\vec w_I$ and $\vec w_O$ are the restrictions of $\vec w$ to the inputs and outputs of $D$, respectively.
\end{theorem}
\begin{proof}
  We can show this by ``firing'' the Pauli web $\vec w$ (see Example~\ref{ex:pauli-web-fire} below). That is, we can treat the overall linear map of $D$ as a tensor product of local states $|\nu\>$, where we use ``caps'' $\<00| + \<11|$ to form internal wires and inputs. We can then replace each local state $|\nu\>$ with $\lambda \vec w_{\nu} |\nu\>$. On interior wires, Pauli operators cancel out (possibly up to a sign, since $Y^T = -Y$), leaving just $\vec w_I$ on the inputs and $\vec w_O$ on the outputs. The sign $\sigma$ is computed as the product of all of the local eigenvalues $\lambda$, multiplied by $-1$ for each $Y$ on an interior or input wire of $\vec w$.
\end{proof}

\begin{example}[Firing a Pauli web]\label{ex:pauli-web-fire}
  Consider the following Pauli web on a diagram $D$. To ``fire'' this Pauli web, we add for each spider the corresponding Paulis around it determined by the colouring on that wire:
  
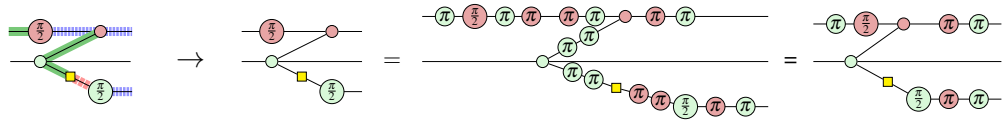
\[
  \begin{tikzpicture}
	\begin{pgfonlayer}{nodelayer}
		\node [style=none] (0) at (-2, 1) {};
		\node [style=none] (1) at (-2, 0) {};
		\node [style=none] (2) at (2, 1) {};
		\node [style=none] (3) at (2, 0) {};
		\node [style=none] (4) at (2, -1) {};
		\node [style=X phase dot] (5) at (-1, 1) {$\frac\pi2$};
		\node [style=Z phase dot] (6) at (1, -1) {$\frac\pi2$};
		\node [style=X dot] (7) at (1, 1) {};
		\node [style=Z dot] (9) at (-1, 0) {};
		\node [style=hadamard] (10) at (0, -0.5) {};
	\end{pgfonlayer}
	\begin{pgfonlayer}{edgelayer}
		\draw [style=Y Web] (6) to (4.center);
		\draw [style=Z Web] (5) to (0.center);
		\draw [style=Y Web] (7) to (2.center);
		\draw [style=Y Web] (7) to (5);
		\draw (1.center) to (9);
		\draw [style=Z Web] (9) to (10);
		\draw [style=X Web] (10) to (6);
		\draw [style=Z Web] (7) to (9);
		\draw (9) to (3.center);
	\end{pgfonlayer}
\end{tikzpicture} \quad \rightarrow \quad
  \begin{tikzpicture}
	\begin{pgfonlayer}{nodelayer}
		\node [style=none] (0) at (-6, 1) {};
		\node [style=none] (1) at (-6, 0) {};
		\node [style=none] (2) at (-2, 1) {};
		\node [style=none] (3) at (-2, 0) {};
		\node [style=none] (4) at (-2, -1) {};
		\node [style=X phase dot] (5) at (-5, 1) {$\frac\pi2$};
		\node [style=Z phase dot] (6) at (-3, -1) {$\frac\pi2$};
		\node [style=X dot] (7) at (-3, 1) {};
		\node [style=Z dot] (8) at (-5, 0) {};
		\node [style=hadamard] (9) at (-4, -0.5) {};
		\node [style=none] (10) at (-1, 0) {$=$};
		\node [style=none] (11) at (0, 1.5) {};
		\node [style=none] (12) at (0, 0) {};
		\node [style=none] (13) at (11.5, 1.5) {};
		\node [style=none] (14) at (11.5, 0) {};
		\node [style=none] (15) at (11.5, -1.5) {};
		\node [style=X phase dot] (16) at (1.75, 1.5) {$\frac\pi2$};
		\node [style=Z phase dot] (17) at (8.75, -1.5) {$\frac\pi2$};
		\node [style=X dot] (18) at (6.75, 1.5) {};
		\node [style=Z dot] (19) at (4, 0) {};
		\node [style=hadamard] (20) at (6.45, -0.875) {};
		\node [style=Z tiny dot] (21) at (0.75, 1.5) {$\pi$};
		\node [style=Z tiny dot] (22) at (2.75, 1.5) {$\pi$};
		\node [style=X tiny dot] (23) at (3.625, 1.5) {$\pi$};
		\node [style=X tiny dot] (24) at (7.75, 1.5) {$\pi$};
		\node [style=Z tiny dot] (25) at (4.825, 0.475) {$\pi$};
		\node [style=Z tiny dot] (26) at (5, -0.35) {$\pi$};
		\node [style=X tiny dot] (27) at (7.2, -1.075) {$\pi$};
		\node [style=X tiny dot] (28) at (9.75, -1.5) {$\pi$};
		\node [style=Z tiny dot] (29) at (10.75, -1.5) {$\pi$};
		\node [style=Z tiny dot] (63) at (5.75, 1.5) {$\pi$};
		\node [style=X tiny dot] (64) at (4.875, 1.5) {$\pi$};
		\node [style=Z tiny dot] (65) at (8.75, 1.5) {$\pi$};
		\node [style=Z tiny dot] (66) at (5.55, 0.85) {$\pi$};
		\node [style=Z tiny dot] (67) at (5.7, -0.6) {$\pi$};
		\node [style=X tiny dot] (68) at (7.925, -1.3) {$\pi$};
		\node [style=none] (69) at (13, 1.25) {};
		\node [style=none] (70) at (13, 0) {};
		\node [style=none] (71) at (19.25, 1.25) {};
		\node [style=none] (72) at (19.25, 0) {};
		\node [style=none] (73) at (19.25, -1.25) {};
		\node [style=X phase dot] (74) at (14.75, 1.25) {$\frac\pi2$};
		\node [style=Z phase dot] (75) at (16.5, -1.25) {$\frac\pi2$};
		\node [style=X dot] (76) at (16, 1.25) {};
		\node [style=Z dot] (77) at (14.25, 0) {};
		\node [style=hadamard] (78) at (15.475, -0.675) {};
		\node [style=Z tiny dot] (79) at (13.75, 1.25) {$\pi$};
		\node [style=X tiny dot] (82) at (17.5, 1.25) {$\pi$};
		\node [style=X tiny dot] (86) at (17.5, -1.25) {$\pi$};
		\node [style=Z tiny dot] (87) at (18.5, -1.25) {$\pi$};
		\node [style=Z tiny dot] (90) at (18.5, 1.25) {$\pi$};
		\node [style=none] (94) at (12.25, 0) {=};
	\end{pgfonlayer}
	\begin{pgfonlayer}{edgelayer}
		\draw (6) to (4.center);
		\draw (5) to (0.center);
		\draw (7) to (2.center);
		\draw (7) to (5);
		\draw (1.center) to (8);
		\draw (8) to (9);
		\draw (9) to (6);
		\draw (7) to (8);
		\draw (8) to (3.center);
		\draw (17) to (15.center);
		\draw (16) to (11.center);
		\draw (18) to (13.center);
		\draw (18) to (16);
		\draw (12.center) to (19);
		\draw (19) to (20);
		\draw (20) to (17);
		\draw (18) to (19);
		\draw (19) to (14.center);
		\draw (75) to (73.center);
		\draw (74) to (69.center);
		\draw (76) to (71.center);
		\draw (76) to (74);
		\draw (70.center) to (77);
		\draw (77) to (78);
		\draw (78) to (75);
		\draw (76) to (77);
		\draw (77) to (72.center);
	\end{pgfonlayer}
\end{tikzpicture}
  \]
  On each internal wire, each Pauli now appears twice so it cancels out. The only Paulis are $\vec P$ on inputs and $\vec Q$ on outputs, hence we obtain $D \propto \vec Q D \vec P$.
\end{example}

Definition~\ref{def:pauli-web} is about eigenvalues, which is not always the easiest way to think of Pauli webs.
It is sometimes useful to unpack the Pauli web criteria into more combinatoric conditions. We will say support of the \textit{same type} to mean Z-support for Z-spiders and X-support for X-spiders, and we will say support for the \textit{opposite type} to mean X-support for Z-spiders and Z-support for X-spiders. It is straightforward to check that following 4 conditions characterise the Pauli operators for which a node $|\nu\>$ is an eigenstate.

\begin{theorem}\label{thm:pauli-web-conds}
A Pauli operator $\vec w \in \mathcal P_{|W|}$ is a Pauli web if and only if it satisfies the following criteria:
\begin{enumerate}
  \item \textbf{H:} every H-gate has a neighbourhood coloured from the set $\{ II, XZ, ZX, YY \}$;
  \item \textbf{all-or-nothing:} the neighbourhood of every spider must either have empty or full support of the opposite type (meaning either all the wires are coloured in the opposing colour, or none of them are);
  \item \textbf{parity:} the neighbourhood of all spiders must have even support of the same type, unless their phase is an odd multiple of $\frac\pi2$, in which case it has even support of the same type iff it has no support of the opposite type; and
  \item \textbf{Clifford:} if the neighbourhood of a spider has support of the opposite type, it must be Clifford.
\end{enumerate}
\end{theorem}


It is clear from Definition~\ref{def:pauli-web} that the product of Pauli webs is again a Pauli web. For Clifford ZX-diagrams, it is often therefore useful to find a minimal generating set of webs. The combinatoric criteria from Theorem~\ref{thm:pauli-web-conds} can be expressed as $\mathbb F_2$-linear equations over $2|W|$ boolean variables, indicating the Z-support and X-support at each wire in $W$. Hence, a generating set of Pauli webs can be computed efficiently by computing a generating set of solutions to a system of linear equations.

We can divide Pauli webs into four categories:
(i) \textit{detectors}, which have no support on inputs or outputs,
(ii) \textit{stabilisers}, which only have support on outputs,
(iii) \textit{costabilisers}, which only have support on inputs, and
(iv) \textit{logicals}, which have support on both inputs and outputs.
All four types play a role in quantum error correction and fault-tolerant quantum computing (see e.g.~\cite{rodatz2025fault,bombin2023unifying}). For our purposes, we will be most interested in logical Pauli webs, which for isometries, take a special form.

\begin{theorem}
  Any Clifford diagram $D$ is an isometry, up to a scalar, if and only if its logical Pauli webs are generated by Pauli webs $\{\ell_Z(i)\}_{i\in I}$ and $\{\ell_X(i)\}_{i\in I}$, where $\ell_Z(i)$ (resp.~$\ell_X(i)$) has $Z$-support (resp. $X$-support) on input $i$ and no other support on the inputs of $D$.
\end{theorem}

\begin{proof}
By Theorem~\ref{thm:pauli-web-stabilisers}, the existence of the Pauli web $\ell_Z(i)$ implies that $DZ_i = \vec {\mathcal Z}_i D$ for some Pauli $\vec {\mathcal Z}_i$, and similarly $\ell_X(i)$ implies $DX_i = \vec{\mathcal X}_i D$ for some $\vec{\mathcal X}_i$. The Paulis $\vec{\mathcal Z}_i$, $\vec{\mathcal X}_i$ then give a complete stabiliser tableau for $D$, which implies that it is an isometry (these Paulis $\vec{\mathcal Z}_i$, $\vec{\mathcal X}_i$ must be independent in $\mathcal{P}_n$ since otherwise $\ell_Z(i)$ and $\ell_X(i)$ wouldn't be independent generators of the logical Pauli webs).
For the converse direction we find the logical operators of the isometry, and construct the corresponding Pauli webs by observing how we should push the Paulis from input to outputs.
\end{proof}

If $D$ is an isometry with $n$ inputs and $n$ outputs, then it is a unitary. Hence, it is completely fixed by its $2n$ logical Pauli webs. If it has $k < n$ inputs, then it is fixed by $2k$ logical webs plus $(n + k) - 2k = n - k$ additional stabilising Pauli webs. These correspond precisely to the stabilisers and logical operators of a quantum error correcting code.

\begin{example}[Pauli webs for a unitary]
  The following 3 qubit unitary has six independent logical Pauli webs, which correspond to its stabiliser tableau:
  \ctikzfig{logical-web-ex}
\end{example}

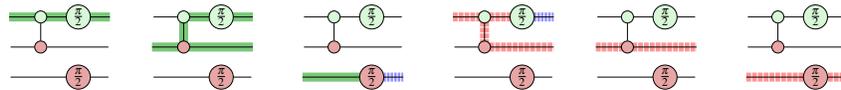
\begin{example}[Pauli webs for an isometry]
  The following isometry is the embedding map for the $[[4, 2, 2]]$ error detecting code. It has 4 logical Pauli webs and 2 stabiliser Pauli webs, corresponding to the logical operators and stabilisers of that code:
  
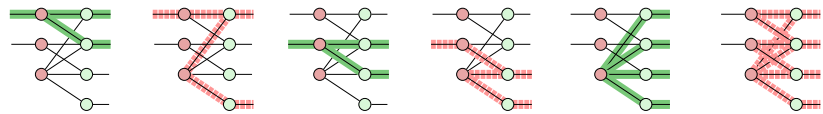
\[
  \begin{tikzpicture}
	\begin{pgfonlayer}{nodelayer}
		\node [style=none] (17) at (-8.5, 1.5) {};
		\node [style=none] (18) at (-8.5, 0.5) {};
		\node [style=X dot] (19) at (-7.5, -0.5) {};
		\node [style=none] (20) at (-5.25, 1.5) {};
		\node [style=none] (21) at (-5.25, 0.5) {};
		\node [style=none] (22) at (-5.25, -0.5) {};
		\node [style=none] (23) at (-5.25, -1.5) {};
		\node [style=Z dot] (24) at (-6, 1.5) {};
		\node [style=Z dot] (25) at (-6, 0.5) {};
		\node [style=Z dot] (26) at (-6, -0.5) {};
		\node [style=Z dot] (27) at (-6, -1.5) {};
		\node [style=X dot] (28) at (-7.5, 1.5) {};
		\node [style=X dot] (29) at (-7.5, 0.5) {};
		\node [style=none] (30) at (0.75, 1.5) {};
		\node [style=none] (31) at (0.75, 0.5) {};
		\node [style=X dot] (32) at (1.75, -0.5) {};
		\node [style=none] (33) at (4, 1.5) {};
		\node [style=none] (34) at (4, 0.5) {};
		\node [style=none] (35) at (4, -0.5) {};
		\node [style=none] (36) at (4, -1.5) {};
		\node [style=Z dot] (37) at (3.25, 1.5) {};
		\node [style=Z dot] (38) at (3.25, 0.5) {};
		\node [style=Z dot] (39) at (3.25, -0.5) {};
		\node [style=Z dot] (40) at (3.25, -1.5) {};
		\node [style=X dot] (41) at (1.75, 1.5) {};
		\node [style=X dot] (42) at (1.75, 0.5) {};
		\node [style=none] (43) at (-3.75, 1.5) {};
		\node [style=none] (44) at (-3.75, 0.5) {};
		\node [style=X dot] (45) at (-2.75, -0.5) {};
		\node [style=none] (46) at (-0.5, 1.5) {};
		\node [style=none] (47) at (-0.5, 0.5) {};
		\node [style=none] (48) at (-0.5, -0.5) {};
		\node [style=none] (49) at (-0.5, -1.5) {};
		\node [style=Z dot] (50) at (-1.25, 1.5) {};
		\node [style=Z dot] (51) at (-1.25, 0.5) {};
		\node [style=Z dot] (52) at (-1.25, -0.5) {};
		\node [style=Z dot] (53) at (-1.25, -1.5) {};
		\node [style=X dot] (54) at (-2.75, 1.5) {};
		\node [style=X dot] (55) at (-2.75, 0.5) {};
		\node [style=none] (56) at (5.5, 1.5) {};
		\node [style=none] (57) at (5.5, 0.5) {};
		\node [style=X dot] (58) at (6.5, -0.5) {};
		\node [style=none] (59) at (8.75, 1.5) {};
		\node [style=none] (60) at (8.75, 0.5) {};
		\node [style=none] (61) at (8.75, -0.5) {};
		\node [style=none] (62) at (8.75, -1.5) {};
		\node [style=Z dot] (63) at (8, 1.5) {};
		\node [style=Z dot] (64) at (8, 0.5) {};
		\node [style=Z dot] (65) at (8, -0.5) {};
		\node [style=Z dot] (66) at (8, -1.5) {};
		\node [style=X dot] (67) at (6.5, 1.5) {};
		\node [style=X dot] (68) at (6.5, 0.5) {};
	\end{pgfonlayer}
	\begin{pgfonlayer}{edgelayer}
		\draw (27) to (23.center);
		\draw (22.center) to (26);
		\draw [style=Z Web] (25) to (21.center);
		\draw [style=Z Web] (24) to (20.center);
		\draw (19) to (24);
		\draw (19) to (25);
		\draw (26) to (19);
		\draw (19) to (27);
		\draw [style=Z Web] (17.center) to (28);
		\draw (18.center) to (29);
		\draw [style=Z Web] (28) to (24);
		\draw [style=Z Web] (28) to (25);
		\draw (29) to (25);
		\draw (29) to (26);
		\draw (40) to (36.center);
		\draw [style=Z Web] (35.center) to (39);
		\draw [style=Z Web] (38) to (34.center);
		\draw (37) to (33.center);
		\draw (32) to (37);
		\draw (32) to (38);
		\draw (39) to (32);
		\draw (32) to (40);
		\draw (30.center) to (41);
		\draw [style=Z Web] (31.center) to (42);
		\draw (41) to (37);
		\draw (41) to (38);
		\draw [style=Z Web] (42) to (38);
		\draw [style=Z Web] (42) to (39);
		\draw [style=X Web] (53) to (49.center);
		\draw (48.center) to (52);
		\draw (51) to (47.center);
		\draw [style=X Web] (50) to (46.center);
		\draw [style=X Web] (45) to (50);
		\draw (45) to (51);
		\draw (52) to (45);
		\draw [style=X Web] (45) to (53);
		\draw [style=X Web] (43.center) to (54);
		\draw (44.center) to (55);
		\draw [style=X Web] (54) to (50);
		\draw (54) to (51);
		\draw (55) to (51);
		\draw (55) to (52);
		\draw [style=X Web] (66) to (62.center);
		\draw [style=X Web] (61.center) to (65);
		\draw (64) to (60.center);
		\draw (63) to (59.center);
		\draw (58) to (63);
		\draw (58) to (64);
		\draw [style=X Web] (65) to (58);
		\draw [style=X Web] (58) to (66);
		\draw (56.center) to (67);
		\draw [style=X Web] (57.center) to (68);
		\draw (67) to (63);
		\draw (67) to (64);
		\draw (68) to (64);
		\draw [style=X Web] (68) to (65);
	\end{pgfonlayer}
\end{tikzpicture} \quad
  \begin{tikzpicture}
	\begin{pgfonlayer}{nodelayer}
		\node [style=none] (0) at (-5, 1.5) {};
		\node [style=none] (1) at (-5, 0.5) {};
		\node [style=X dot] (2) at (-4, -0.5) {};
		\node [style=none] (3) at (-1.75, 1.5) {};
		\node [style=none] (4) at (-1.75, 0.5) {};
		\node [style=none] (5) at (-1.75, -0.5) {};
		\node [style=none] (6) at (-1.75, -1.5) {};
		\node [style=Z dot] (7) at (-2.5, 1.5) {};
		\node [style=Z dot] (8) at (-2.5, 0.5) {};
		\node [style=Z dot] (9) at (-2.5, -0.5) {};
		\node [style=Z dot] (10) at (-2.5, -1.5) {};
		\node [style=X dot] (11) at (-4, 1.5) {};
		\node [style=X dot] (12) at (-4, 0.5) {};
		\node [style=none] (13) at (0, 1.5) {};
		\node [style=none] (14) at (0, 0.5) {};
		\node [style=X dot] (15) at (1, -0.5) {};
		\node [style=none] (16) at (3.25, 1.5) {};
		\node [style=none] (17) at (3.25, 0.5) {};
		\node [style=none] (18) at (3.25, -0.5) {};
		\node [style=none] (19) at (3.25, -1.5) {};
		\node [style=Z dot] (20) at (2.5, 1.5) {};
		\node [style=Z dot] (21) at (2.5, 0.5) {};
		\node [style=Z dot] (22) at (2.5, -0.5) {};
		\node [style=Z dot] (23) at (2.5, -1.5) {};
		\node [style=X dot] (24) at (1, 1.5) {};
		\node [style=X dot] (25) at (1, 0.5) {};
	\end{pgfonlayer}
	\begin{pgfonlayer}{edgelayer}
		\draw [style=Z Web] (10) to (6.center);
		\draw [style=Z Web] (5.center) to (9);
		\draw [style=Z Web] (8) to (4.center);
		\draw [style=Z Web] (7) to (3.center);
		\draw [style=Z Web] (2) to (7);
		\draw [style=Z Web] (2) to (8);
		\draw [style=Z Web] (9) to (2);
		\draw [style=Z Web] (2) to (10);
		\draw (0.center) to (11);
		\draw (1.center) to (12);
		\draw (11) to (7);
		\draw (11) to (8);
		\draw (12) to (8);
		\draw (12) to (9);
		\draw [style=X Web] (23) to (19.center);
		\draw [style=X Web] (18.center) to (22);
		\draw [style=X Web] (21) to (17.center);
		\draw [style=X Web] (20) to (16.center);
		\draw [style=X Web] (15) to (20);
		\draw [style=X Web] (15) to (21);
		\draw [style=X Web] (22) to (15);
		\draw [style=X Web] (15) to (23);
		\draw (13.center) to (24);
		\draw (14.center) to (25);
		\draw [style=X Web] (24) to (20);
		\draw [style=X Web] (24) to (21);
		\draw [style=X Web] (25) to (21);
		\draw [style=X Web] (25) to (22);
	\end{pgfonlayer}
\end{tikzpicture}
  \]
\end{example}

\subsection{Pauli Flow}\label{sec:Pauli-flow}

The standard flow conditions (causal flow, generalised flow, Pauli flow) were originally defined for graph states, but they also make sense for \emph{graph-like} ZX-diagrams~\cite{duncan2019graph,Backens2020extraction}.

\begin{definition}
A ZX-diagram is called \textit{graph-like} if:
\begin{enumerate}
 \item every spider is of type $Z$ and internal wires only connect $Z$-spiders to $H$-gates;
 \item no $H$-gate is connected to the same spider via both adjacent wires (no self-loops) and no two $H$-gates are connected to the same pair of spiders (no parallel edges); and
 \item every input and output is connected to a spider via a plain edge, and every spider is connected to at most 1 input and at most 1 output.
\end{enumerate}
\end{definition}

Graph-like diagrams are those ZX-diagrams $D$ whose structure is uniquely fixed by an \textit{open graph} $G_D$. An open graph is a simple, undirected graph with chosen (not necessarily disjoint) subsets of input and output vertices. $G_D$ is the open graph whose vertices are $Z$-spiders and whose edges are connections between spiders via an $H$-gate. Input vertices are the set of spiders adjacent to an input wire and output vertices are the set of spiders adjacent to an output wire. We will represent the $H$-edges using dashed blue lines following~\cite{duncan2019graph}.

An open graph, along with an ordering of its vertices and a choice of measurement angles, defines part of a measurement-based quantum computation (MBQC) in the one-way model~\cite{browne2007gflow}. In MBQC, one prepares a type of stabiliser state called a \textit{graph state} and performs single-qubit measurements, which can be either Pauli measurements or measurements within one of the three planes $XY$, $YZ$, or $XZ$ of the Bloch sphere.
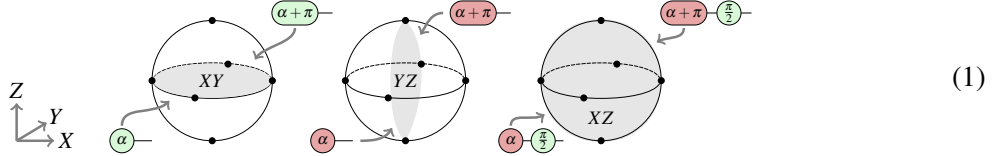
\begin{equation}\label{eq:bloch-planes}
  \begin{tikzpicture}
	\begin{pgfonlayer}{nodelayer}
		\node [style=none] (0) at (-6.5, 2) {};
		\node [style=none] (1) at (-8.5, 0) {};
		\node [style=none] (2) at (-4.5, 0) {};
		\node [style=none] (3) at (-6.5, -2) {};
		\node [style=vertex] (4) at (-6.5, 2) {};
		\node [style=vertex] (5) at (-8.5, 0) {};
		\node [style=vertex] (6) at (-6.5, -2) {};
		\node [style=vertex] (7) at (-4.5, 0) {};
		\node [style=vertex] (8) at (-5.975, 0.55) {};
		\node [style=vertex] (9) at (-7.075, -0.575) {};
		\node [style=none] (10) at (-13, -2) {};
		\node [style=none] (11) at (-13, -0.75) {};
		\node [style=none] (12) at (-11.75, -2) {};
		\node [style=none] (13) at (-12, -1.4) {};
		\node [style=none] (14) at (-13, -0.325) {$Z$};
		\node [style=none] (15) at (-11.325, -2) {$X$};
		\node [style=none] (16) at (-11.675, -1.2) {$Y$};
		\node [style=none] (19) at (-8.5, 0) {};
		\node [style=none] (20) at (-4.5, 0) {};
		\node [style=none] (21) at (-6.5, 0) {\footnotesize $XY$};
		\node [style=none] (22) at (-8.5, -2) {};
		\node [style=Z phase dot] (23) at (-9.5, -2) {$\alpha$};
		\node [style=none] (24) at (-2.5, 2.25) {};
		\node [style=Z phase dot] (25) at (-3.75, 2.25) {$\ \alpha+\pi~$};
		\node [style=none] (26) at (-9.5, -1.5) {};
		\node [style=none] (27) at (-7.925, -0.675) {};
		\node [style=none] (28) at (-3.75, 1.75) {};
		\node [style=none] (29) at (-5.075, 0.7) {};
		\node [style=none] (30) at (-0.075, 2) {};
		\node [style=none] (31) at (-2.075, 0) {};
		\node [style=none] (32) at (1.925, 0) {};
		\node [style=none] (33) at (-0.075, -2) {};
		\node [style=vertex] (34) at (-0.075, 2) {};
		\node [style=vertex] (35) at (-2.075, 0) {};
		\node [style=vertex] (36) at (-0.075, -2) {};
		\node [style=vertex] (37) at (1.925, 0) {};
		\node [style=vertex] (38) at (0.45, 0.55) {};
		\node [style=vertex] (39) at (-0.65, -0.575) {};
		\node [style=none] (40) at (-0.075, 2) {};
		\node [style=none] (41) at (-0.075, -2) {};
		\node [style=none] (42) at (-0.075, 0) {\footnotesize $YZ$};
		\node [style=none] (43) at (-1.825, -2) {};
		\node [style=X phase dot] (44) at (-2.825, -2) {$\alpha$};
		\node [style=none] (45) at (3.425, 2.25) {};
		\node [style=X phase dot] (46) at (2.175, 2.25) {$\ \alpha+\pi~$};
		\node [style=none] (47) at (-1.575, -2) {};
		\node [style=none] (48) at (-0.525, -1.575) {};
		\node [style=none] (49) at (1.175, 2.25) {};
		\node [style=none] (50) at (0.425, 1.525) {};
		\node [style=none] (51) at (6.425, 2) {};
		\node [style=none] (52) at (4.425, 0) {};
		\node [style=none] (53) at (8.425, 0) {};
		\node [style=none] (54) at (6.425, -2) {};
		\node [style=vertex] (55) at (6.425, 2) {};
		\node [style=vertex] (56) at (4.425, 0) {};
		\node [style=vertex] (57) at (6.425, -2) {};
		\node [style=vertex] (58) at (8.425, 0) {};
		\node [style=vertex] (59) at (6.95, 0.55) {};
		\node [style=vertex] (60) at (5.85, -0.575) {};
		\node [style=none] (61) at (4.425, 0) {};
		\node [style=none] (62) at (8.425, 0) {};
		\node [style=none] (63) at (6.425, -1.25) {\footnotesize $XZ$};
		\node [style=none] (64) at (5.175, -2) {};
		\node [style=X phase dot] (65) at (3.425, -2) {$\alpha$};
		\node [style=none] (66) at (11.425, 2.25) {};
		\node [style=X phase dot] (67) at (9.175, 2.25) {$\ \alpha+\pi~$};
		\node [style=none] (68) at (3.425, -1.5) {};
		\node [style=none] (69) at (4.55, -1.125) {};
		\node [style=none] (70) at (9.175, 1.75) {};
		\node [style=none] (71) at (8.275, 1.275) {};
		\node [style=Z phase dot] (72) at (4.5, -2) {$\frac\pi2$};
		\node [style=Z phase dot] (73) at (10.675, 2.25) {$\frac\pi2$};
	\end{pgfonlayer}
	\begin{pgfonlayer}{edgelayer}
		\draw [style=light-fill] (41.center)
			 to [in=-150, out=150, looseness=0.50] (40.center)
			 to [in=60, out=-45, looseness=0.75] cycle;
		\draw [style=light-fill] (62.center)
			 to [in=-90, out=-90, looseness=1.75] (61.center)
			 to [in=90, out=90, looseness=1.75] cycle;
		\draw [bend right=45] (0.center) to (1.center);
		\draw [bend right=45] (1.center) to (3.center);
		\draw [bend right=45] (3.center) to (2.center);
		\draw [bend right=45] (2.center) to (0.center);
		\draw [style=dashed edge, in=120, out=60, looseness=0.50] (5) to (7);
		\draw [in=-120, out=-60, looseness=0.50] (5) to (7);
		\draw [style=pointer edge] (10.center) to (11.center);
		\draw [style=pointer edge] (10.center) to (12.center);
		\draw [style=pointer edge] (10.center) to (13.center);
		\draw [style=light-fill] (20.center)
			 to [in=-75, out=-105, looseness=0.50] (19.center)
			 to [in=105, out=75, looseness=0.50] cycle;
		\draw (23) to (22.center);
		\draw (25) to (24.center);
		\draw [style=pointer edge, in=-135, out=90] (26.center) to (27.center);
		\draw [style=pointer edge, in=60, out=-90, looseness=0.75] (28.center) to (29.center);
		\draw [bend right=45] (30.center) to (31.center);
		\draw [bend right=45] (31.center) to (33.center);
		\draw [bend right=45] (33.center) to (32.center);
		\draw [bend right=45] (32.center) to (30.center);
		\draw [style=dashed edge, in=120, out=60, looseness=0.50] (35) to (37);
		\draw [in=-120, out=-60, looseness=0.50] (35) to (37);
		\draw (44) to (43.center);
		\draw (46) to (45.center);
		\draw [style=pointer edge, in=-120, out=0] (47.center) to (48.center);
		\draw [style=pointer edge, in=60, out=180] (49.center) to (50.center);
		\draw [bend right=45] (51.center) to (52.center);
		\draw [bend right=45] (52.center) to (54.center);
		\draw [bend right=45] (54.center) to (53.center);
		\draw [bend right=45] (53.center) to (51.center);
		\draw [style=dashed edge, in=120, out=60, looseness=0.50] (56) to (58);
		\draw [in=-120, out=-60, looseness=0.50] (56) to (58);
		\draw (65) to (64.center);
		\draw (67) to (66.center);
		\draw [style=pointer edge, in=-150, out=90] (68.center) to (69.center);
		\draw [style=pointer edge, in=30, out=-90] (70.center) to (71.center);
	\end{pgfonlayer}
\end{tikzpicture}
\end{equation}

Graph-like ZX-diagrams can be interpreted as graph states, where each of the non-output vertices is being measured in a basis determined by its phase, e.g.
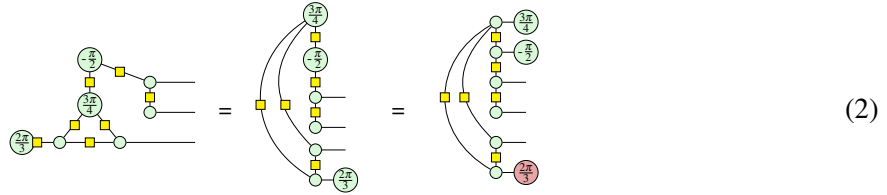
\begin{equation}\label{eq:graph-state-ex}
\begin{tikzpicture}
	\begin{pgfonlayer}{nodelayer}
		\node [style=none] (0) at (-1.25, 1) {};
		\node [style=none] (2) at (-1.25, 0) {};
		\node [style=none] (3) at (-1.25, -1) {};
		\node [style=Z dot] (5) at (-2.75, 1) {};
		\node [style=Z dot] (6) at (-3.75, -1) {};
		\node [style=Z dot] (7) at (-2.75, 0) {};
		\node [style=Z phase dot] (11) at (-4.75, 0.25) {$\frac{3\pi}{4}$};
		\node [style=Z phase dot] (13) at (-7, -1) {$\frac{2\pi}{3}$};
		\node [style=Z phase dot] (14) at (-4.75, 1.75) {$\text{-}\frac\pi2$};
		\node [style=Z dot] (15) at (-5.75, -1) {};
		\node [style=none] (17) at (-0.25, 0) {=};
		\node [style=none] (18) at (3.75, 0.5) {};
		\node [style=none] (19) at (3.75, -0.5) {};
		\node [style=none] (20) at (3.75, -1.25) {};
		\node [style=Z dot] (21) at (2.75, 0.5) {};
		\node [style=Z dot] (22) at (2.75, -1.25) {};
		\node [style=Z dot] (23) at (2.75, -0.5) {};
		\node [style=Z phase dot] (24) at (2.75, 3.25) {$\frac{3\pi}{4}$};
		\node [style=Z phase dot] (25) at (3.75, -2.25) {$\frac{2\pi}{3}$};
		\node [style=Z phase dot] (26) at (2.75, 1.75) {$\text{-}\frac\pi2$};
		\node [style=Z dot] (27) at (2.75, -2.25) {};
		\node [style=none] (28) at (5.5, 0) {=};
		\node [style=none] (29) at (9.75, 1) {};
		\node [style=none] (30) at (9.75, 0) {};
		\node [style=none] (31) at (9.75, -1) {};
		\node [style=Z dot] (32) at (8.75, 1) {};
		\node [style=Z dot] (33) at (8.75, -1) {};
		\node [style=Z dot] (34) at (8.75, 0) {};
		\node [style=Z dot] (35) at (8.75, 3) {};
		\node [style=X phase dot] (36) at (9.75, -2) {$\frac{2\pi}{3}$};
		\node [style=Z dot] (37) at (8.75, 2) {};
		\node [style=Z dot] (38) at (8.75, -2) {};
		\node [style=Z phase dot] (40) at (9.75, 3) {$\frac{3\pi}{4}$};
		\node [style=Z phase dot] (41) at (9.75, 2) {$\text{-}\frac\pi2$};
		\node [style=hadamard] (42) at (-6.5, -1) {};
		\node [style=hadamard] (43) at (-4.75, -1) {};
		\node [style=hadamard] (44) at (-4.25, -0.425) {};
		\node [style=hadamard] (45) at (-5.25, -0.425) {};
		\node [style=hadamard] (46) at (-3.75, 1.35) {};
		\node [style=hadamard] (47) at (-4.75, 1) {};
		\node [style=hadamard] (48) at (-2.75, 0.5) {};
		\node [style=hadamard] (49) at (0.925, 0.25) {};
		\node [style=hadamard] (50) at (1.75, 0.25) {};
		\node [style=hadamard] (51) at (2.75, 1) {};
		\node [style=hadamard] (52) at (2.75, 2.5) {};
		\node [style=hadamard] (53) at (2.75, 0) {};
		\node [style=hadamard] (54) at (7.05, 0.5) {};
		\node [style=hadamard] (55) at (7.7, 0.5) {};
		\node [style=hadamard] (56) at (8.75, 0.5) {};
		\node [style=hadamard] (57) at (8.75, 1.5) {};
		\node [style=hadamard] (58) at (8.75, 2.5) {};
		\node [style=hadamard] (59) at (8.75, -1.5) {};
		\node [style=hadamard] (60) at (2.75, -1.75) {};
	\end{pgfonlayer}
	\begin{pgfonlayer}{edgelayer}
		\draw (0.center) to (5);
		\draw (7) to (2.center);
		\draw (6) to (3.center);
		\draw (13) to (15);
		\draw (11) to (6);
		\draw (14) to (5);
		\draw (14) to (11);
		\draw (15) to (11);
		\draw (15) to (6);
		\draw (18.center) to (21);
		\draw (23) to (19.center);
		\draw (22) to (20.center);
		\draw (25) to (27);
		\draw (21) to (23);
		\draw [bend right=45, looseness=1.25] (24) to (22);
		\draw (26) to (21);
		\draw (26) to (24);
		\draw [bend left=60, looseness=1.25] (27) to (24);
		\draw (27) to (22);
		\draw (29.center) to (32);
		\draw (34) to (30.center);
		\draw (33) to (31.center);
		\draw (36) to (38);
		\draw (32) to (34);
		\draw [bend right=45, looseness=1.25] (35) to (33);
		\draw (37) to (32);
		\draw (37) to (35);
		\draw [bend left=60, looseness=1.25] (38) to (35);
		\draw (38) to (33);
		\draw (35) to (40);
		\draw (41) to (37);
		\draw (5) to (7);
	\end{pgfonlayer}
\end{tikzpicture}
\end{equation}

While individual measurement outcomes are non-deterministic, one can recover deterministic computation via a mechanism called \textit{feed forward}, where later measurement choices are adapted based on intermediate measurement outcomes. We will only very roughly sketch the one-way model and its relationship to the ZX-calculus here, referring the interested reader to e.g.~\cite{Backens2020extraction} or \cite[Chapter~9]{KissingerWetering2024Book} for a more in-depth treatment.

A generic way to define a correction strategy is to define a \textit{correction set} $c(v)$ for each non-output vertex $v$. Let $\overline I := V_G \backslash I$ and $\overline O := V_G \backslash O$ be the \textit{non-inputs} and \textit{non-outputs}, respectively. Let $c : \overline O \to 2^{\overline I}$ be defined as a function from non-output vertices to correction sets.

Similar to a Pauli web, a correction set can be ``fired'' on a graph state to introduce a stabiliser for that graph state. The stabilisers of graph states are all of the form $S_v := X_v \Pi_{w \in \mathcal N(v)} Z_w$. That is, they introduce a Pauli $X$ on a given vertex and a Pauli $Z$ on all neighbours of that vertex. We can ``fire'' a correction set $c(v)$ by introducing the stabiliser $\Pi_{u \in c(v)} S_u$. This will introduce an $X$ on all the vertices in the set $c(v)$ and introduce a $Z$ on all the vertices in the \textit{odd neighbourhood} $\Odd(c(v))$, i.e. the set of all vertices connected to $c(v)$ oddly many times.

Now, we want to arrange our $c(v)$ such that ``firing'' $c(v)$ will flip (i.e. correct) the measurement outcome at $v$ while only disturbing other measurements in the future of $v$. How we choose $c(v)$ depends on what kind of measurement we are doing at $v$, which is dictated by a measurement label $\mu(v) \in \{ X, Y, Z, XY, YZ, XZ \}$. For convenience, we will treat these six measurement labels as 1- and 2-element sets, writing e.g. $X \in \mu(v)$ to mean $\mu(v) = X, XY,$ or $XZ$.

A choice of correction sets accounting for all six of these types of measurements is called \textit{Pauli flow}. Normally, this is stated as nine distinct conditions. However, we can treat all of the measurement choices symmetrically if we define the following three \textit{anti-commutation sets}:
\[
A_X(u) = \Odd(c(u))
\qquad
A_Y(u) = c(u) \Delta \Odd(c(u))
\qquad
A_Z(u) = c(u)
\]
i.e. locations which get a Pauli that anti-commutes with the given Pauli when the stabiliser $\Pi_{u \in c(v)} S_u$ is introduced.

\begin{definition}\label{def:pauli-flow}
  An open graph $(G, I, O)$ with a choice of measurements $\mu$ has \textit{Pauli flow} if there exists a pair $(\preceq, c)$ such that:
  \begin{enumerate}
    \item [(i)] $P \in \mu(u) \implies u \in A_P(u)$
    \item [(ii)] $P \in \mu(v), v \in A_P(u) \implies u \preceq v$
  \end{enumerate}
\end{definition}

This is not the standard way Pauli flow is presented in the literature, but it will make the proofs in the following sections much simpler. A proof of equivalence to the standard definition from~\cite{browne2007gflow} is given in Appendix~\ref{app:pauli-flow-def}.

Although Pauli flow is defined for open graphs, we can also interpret it as a condition on graph-like ZX-diagrams. While in general there can be different choices of measurement labels compatible with a ZX-diagram, we will assume measurement labels are chosen as follows:
(i) spiders whose phase is a multiple of $\pi$ have $\mu(\nu) = X$,
(ii) spiders whose phase is an odd multiple of $\pi/2$ have $\mu(\nu) = Y$, and
(iii) all other (necessarily non-Clifford) spiders have $\mu(\nu) = XY$.

As explained in~\cite[Section~4.2]{mcelvanney2025thesis}, spiders with measurement labels in the $YZ$ and $XZ$ planes can always be replaced with pairs of spiders, where one is measured in $X$ and $Y$ and the other in $XY$ (e.g.~we interpret the $\frac\pi2$ phases in~\eqref{eq:bloch-planes} for $XZ$ as additional spiders in the diagram). For this reason, we will only consider ZX-diagrams with measurement labels in the set $X$, $Y$, and $XY$ and we choose as many spiders to be Pauli-measured as possible, based on their angle.


\section{Pauli Semiwebs}\label{sec:semiwebs}

In this section, we introduce a new generalisation of Pauli webs, which allow the web conditions to be relaxed at certain locations, which we call \textit{defects}. This generalisation will allow us to define many Pauli semiwebs on non-Clifford ZX-diagrams.

For a spider $\nu \in N$, we define its \emph{twisted local state} $|\nu_\alpha\>$ as $|0\>^{\otimes k} + e^{i (\phi(\nu) + \alpha)}|1\>^{\otimes k}$ for Z-spiders, $|{+}\>^{\otimes k} + e^{i (\phi(\nu) + \alpha)}|{-}\>^{\otimes k}$ for X-spiders, and $|\nu\>$ for $H$-gates. Hence, $|\nu_0\> = |\nu\>$ and twisting has no effect on $H$-gates. This is just a notational convenience to avoid a case distinction in the following definition.

\begin{definition}
A \textit{Pauli semiweb} for a diagram $D$ is a Pauli operator $\vec w \in \mathcal P_{|W|}$ such that for all nodes $\nu$ there exists an $\alpha$ such that:
\begin{equation}\label{eq:twist}
\vec w_\nu|\nu\> = \lambda |\nu_\alpha\>
\end{equation}
If $\alpha \neq 0$, then $\nu$ is called an $\alpha$-\textit{defect}, or simply a \textit{defect}, of the Pauli semiweb.
\end{definition}

By definition, a Pauli semiweb with no defects is a Pauli web. If a spider has one or more legs, then $\alpha$ is uniquely fixed by \eqref{eq:twist}. Like Pauli webs, we can characterise Pauli semiwebs in terms of the local Z- and X-support around nodes. In fact, the semiweb conditions are a strict subset of the Pauli web conditions from Theorem~\ref{thm:pauli-web-conds}.

\begin{theorem}\label{thm:semiweb-cond}
A Pauli operator $\vec w \in \mathcal P_{|W|}$ is a Pauli semiweb iff it satisfies the \textbf{H} and \textbf{all-or-nothing} conditions from Theorem~\ref{thm:pauli-web-conds}.
\end{theorem}

\begin{proof}
  The behaviour of $H$-gates is the same for Pauli semiwebs, so it suffices to show that the \textbf{all-or-nothing} condition is equivalent to \eqref{eq:twist} for all spiders $\nu$. Let $\nu$ be a Z-spider. Suppose we write $\vec w_{\nu} |\nu\> = \lambda \vec P_X \vec P_Z |\nu\>$, where $\vec P_Z$ contains only $Z$ operators and $\vec P_X$ contains only $X$ operators. Then, the RHS equals $\lambda\vec P_X |\nu_\alpha\>$ for $\alpha = 0$ if the number of $Z$ operators is even and $\alpha = \pi$ if it is odd. Now, \eqref{eq:twist} will be satisfied iff $\vec P_X|\nu_\alpha\>$ is proportional to a Z-spider. This happens if and only if $\vec P_X$ either preserves the basis states $|0...0\>$ and $|1...1\>$ or it swaps them, which is true iff $\vec P_X$ is either $I$ or $X\otimes \ldots \otimes X$. Hence, \eqref{eq:twist} is satisfied for Z-spiders iff their neighbourhood satisfies the all-or-nothing condition. We can argue similarly for X-spiders.
\end{proof}

As a consequence of this theorem, defects are introduced precisely when the \textbf{Clifford} and \textbf{parity} conditions from Theorem~\ref{thm:pauli-web-conds} are violated. If the \textbf{Clifford} condition is violated, we get a $-2\alpha$ defect, if the \textbf{parity} condition is violated, we get a $\pi$ defect, and if both are violated, we get a $\pi - 2\alpha$ defect.
The following result is an immediate consequence of the fact that products of Pauli operators preserve the \textbf{H} and \textbf{all-or-nothing} criteria. 

\begin{corollary}\label{cor:semiweb-product}
The product of Pauli semiwebs is a Pauli semiweb.
\end{corollary}


Because of this, we can also characterise Pauli semiwebs by identifying a generating set of semiwebs.

\begin{definition}\label{def:basic-semiweb}
  For a spider $\nu$ in $D$, the \textit{basic semiweb} $\vec b_\nu$ of $\nu$ is the smallest semiweb that has support of the opposite colour on a wire adjacent to $\nu$. That is, it has X-support if $\nu$ is a Z-spider and it has Z-support if $\nu$ is an X-spider.
\end{definition}

\begin{figure}
  \begin{equation*}
    \begin{tikzpicture}
	\begin{pgfonlayer}{nodelayer}
		\node [style=Z phase dot] (0) at (-0.5, 1) {$\alpha$};
		\node [style=X dot] (1) at (-0.5, 2) {};
		\node [style=X dot] (2) at (-2.5, 1) {};
		\node [style=Z dot] (3) at (-0.5, 0) {};
		\node [style=Z dot] (4) at (1.25, 1) {};
		\node [style=X dot] (5) at (1.25, 0) {};
		\node [style=X dot] (6) at (1.25, -1) {};
		\node [style=hadamard] (7) at (0.5, 1) {};
		\node [style=hadamard] (8) at (-1.5, 1) {};
		\node [style=Z dot] (9) at (-2.5, 0) {};
		\node [style=none] (10) at (2, 0) {};
		\node [style=none] (11) at (2, -1) {};
		\node [style=none] (12) at (2, 1) {};
		\node [style=none] (13) at (-1.75, 0) {};
		\node [style=none] (14) at (0.25, 2) {};
	\end{pgfonlayer}
	\begin{pgfonlayer}{edgelayer}
		\draw [style=Z Web] (8) to (2);
		\draw [style=Z Web] (7) to (4);
		\draw [style=Z Web] (2) to (9);
		\draw [style=X Web] (0) to (8);
		\draw [style=X Web] (0) to (7);
		\draw [style=X Web] (0) to (1);
		\draw [style=X Web] (0) to (3);
		\draw [style=X Web] (3) to (5);
		\draw [style=X Web] (3) to (6);
		\draw (5) to (10.center);
		\draw (6) to (11.center);
		\draw (4) to (12.center);
		\draw (9) to (13.center);
		\draw (1) to (14.center);
	\end{pgfonlayer}
\end{tikzpicture} \qquad\quad \begin{tikzpicture}
	\begin{pgfonlayer}{nodelayer}
		\node [style=Z phase dot] (0) at (1, 1.75) {$\beta$};
		\node [style=X phase dot] (2) at (-1, 1.75) {$\alpha$};
		\node [style=hadamard] (8) at (0, 1.75) {};
		\node [style=none] (9) at (-1.75, 2.25) {};
		\node [style=none] (10) at (-1.75, 1.25) {};
		\node [style=none] (11) at (1.75, 2.25) {};
		\node [style=none] (12) at (1.75, 1.25) {};
		\node [style=Z phase dot] (13) at (1, 0) {$\beta$};
		\node [style=Z phase dot] (14) at (-1, 0) {$\alpha$};
		\node [style=none] (16) at (-1.75, 0.5) {};
		\node [style=none] (17) at (-1.75, -0.5) {};
		\node [style=none] (18) at (1.75, 0.5) {};
		\node [style=none] (19) at (1.75, -0.5) {};
		\node [style=Z phase dot] (20) at (0, -1.25) {$\alpha$};
		\node [style=none] (21) at (1.75, -1.25) {};
		\node [style=none] (22) at (0.25, -2.25) {};
		\node [style=none] (23) at (1.5, -1.5) {};
		\node [style=none] (24) at (-1.5, -2.25) {{\scriptsize input/output}};
	\end{pgfonlayer}
	\begin{pgfonlayer}{edgelayer}
		\draw [style=Z Web] (8) to (2);
		\draw [style=X Web] (0) to (8);
		\draw (10.center) to (2);
		\draw (9.center) to (2);
		\draw (0) to (12.center);
		\draw (0) to (11.center);
		\draw (17.center) to (14);
		\draw (16.center) to (14);
		\draw (13) to (19.center);
		\draw (13) to (18.center);
		\draw [style=Z Web] (13) to (14);
		\draw [style=Z Web] (20) to (21.center);
		\draw [style=pointer edge, in=-90, out=0] (22.center) to (23.center);
	\end{pgfonlayer}
\end{tikzpicture} \qquad \quad \begin{tikzpicture}
	\begin{pgfonlayer}{nodelayer}
		\node [style=Z dot] (0) at (-4, 0) {};
		\node [style=X dot] (1) at (-3, 0) {};
		\node [style=X phase dot] (6) at (-2, -0.5) {$\frac\pi2$};
		\node [style=none] (11) at (-1, -0.5) {};
		\node [style=X dot] (14) at (-2, 0.5) {};
		\node [style=none] (15) at (-5, 0) {};
		\node [style=none] (16) at (-1, 0.5) {};
		\node [style=none] (17) at (0, 0) {=};
		\node [style=Z dot] (18) at (2, 1.25) {};
		\node [style=X dot] (19) at (3, 1.25) {};
		\node [style=X phase dot] (20) at (4, 0.75) {$\frac\pi2$};
		\node [style=none] (21) at (5, 0.75) {};
		\node [style=X dot] (22) at (4, 1.75) {};
		\node [style=none] (23) at (1, 1.25) {};
		\node [style=none] (24) at (5, 1.75) {};
		\node [style=none] (32) at (2.75, 0) {$\cdot$};
		\node [style=Z dot] (33) at (2, -1.25) {};
		\node [style=X dot] (34) at (3, -1.25) {};
		\node [style=X phase dot] (35) at (4, -1.75) {$\frac\pi2$};
		\node [style=none] (36) at (5, -1.75) {};
		\node [style=X dot] (37) at (4, -0.75) {};
		\node [style=none] (38) at (1, -1.25) {};
		\node [style=none] (39) at (5, -0.75) {};
	\end{pgfonlayer}
	\begin{pgfonlayer}{edgelayer}
		\draw [style=Z Web] (0) to (1);
		\draw [style=Y Web] (6) to (11.center);
		\draw [style=Y Web] (1) to (14);
		\draw [style=Z Web] (15.center) to (0);
		\draw [style=Z Web] (14) to (16.center);
		\draw [style=Z Web] (1) to (6);
		\draw [style=Z Web] (18) to (19);
		\draw [style=Z Web] (20) to (21.center);
		\draw [style=Z Web] (19) to (22);
		\draw (23.center) to (18);
		\draw [style=Z Web] (22) to (24.center);
		\draw [style=Z Web] (19) to (20);
		\draw (33) to (34);
		\draw [style=X Web] (35) to (36.center);
		\draw [style=X Web] (34) to (37);
		\draw [style=Z Web] (38.center) to (33);
		\draw (37) to (39.center);
		\draw (34) to (35);
	\end{pgfonlayer}
\end{tikzpicture}
  \end{equation*}
  \caption{Left: Example of a basic semiweb. Note that all the spiders on the boundary of the web have a $\pi$-defect. Middle: different possibilities for an edge semiweb. Right: decomposing a Pauli semiweb into a basic semiweb (top) and edge semiwebs (bottom).}
  \label{fig:semiwebs}
\end{figure}
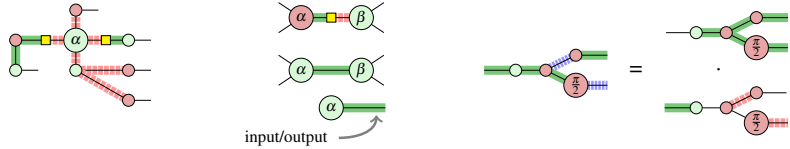

For a graph-like ZX-diagram, the basic semiweb $\vec b_\nu$ at a spider $\nu$ is the Pauli semiweb with an $X$ operator on all of the wires adjacent to $\nu$ and a $Z$ operator on all wires that connect via an $H$-gate to wires adjacent to $\nu$. In other words, it colours all of the $H$-edges, inputs, and outputs adjacent to $\nu$. 
In general, basic semiwebs will colour fusable components of the diagram. We say a pair of spiders is \textit{fusable} if they are either the same colour and connected via a plain edge or opposite colours and connected by an $H$-edge. A fusable component is a connected sub-diagram of fusable spiders. For a spider $\nu$, the basic semiweb $\vec b_\nu$ colours the edges adjacent to every spider in the fusable component of $\nu$. See Figure~\ref{fig:semiwebs} for an example.

Basic semiwebs nearly generate all of the Pauli semiwebs. However, there is one case that they don't cover: if a plain edge or $H$-edge connects a fusable pair of spiders or it contains an input/output wire, we can always form a semiweb consisting just of that edge, with defects at each adjacent spider.

\begin{definition}
  An \textit{edge semiweb} is a Pauli semiweb that colours a single plain edge $X$ or $Z$ or a single $H$-edge $XZ$.
\end{definition}

\begin{theorem}\label{thm:generating-set-of-spiders}
Any Pauli semiweb can be written as $\vec w = \vec e\, \prod_{\nu \in B} \vec b_\nu$ for $\vec e$ a product of edge semiwebs. In this case, we call $B$ a \textit{generating set of spiders} for the Pauli semiweb $\vec w$.
\end{theorem}

\begin{proof}
  Let $\vec w$ be a Pauli semiweb on $D$. It suffices to show that $\vec w$ can be reduced to the identity by multiplying by basic semiwebs and edge semiwebs. Suppose a spider $\nu$ is adjacent to a wire coloured with the opposite type. Then, by minimality of $\vec b_\nu$, the $Z$-support and $X$-support of $\vec w$ must contain the support of $b_\nu$. Hence, multiplying $\vec w$ by $b_\nu$ strictly reduces the support of $\vec w$. We can repeat this until no spider is adjacent to a wire of the opposite colour. The only remaining support will be on plain edges or H-edges whose ends are at inputs/outputs or at spiders of the same type. This can be represented as a product of edge webs.
\end{proof}

Note that when $D$ is graph-like, the set of generating spiders of a semiweb is unique, so we call it \textit{the} generating set of spiders for $\vec w$. See Figure~\ref{fig:semiwebs} for an example of decomposing a Pauli web into basic and edge semiwebs. 


\section{ZX-flow}\label{sec:zx-flow}

Now that we have established some basic concepts for Pauli semiwebs, we are ready to state our main definition.

\begin{definition}
We say a ZX-diagram $D$ has \textit{ZX-flow} if there exists a partial ordering $\preceq \ \subseteq T \times T$ on the non-Clifford spiders $T \subseteq N_D$ and we have the following Pauli semiwebs:
\begin{itemize}
  \item for every input wire $i \in I$, a pair of Pauli semiwebs $\ell_Z(i)$ and $\ell_X(i)$ that only have defects at non-Clifford spiders and whose restrictions to input wires are $Z_i$ and $X_i$, respectively; and
  \item for every non-Clifford spider $\nu \in T$, a Pauli semiweb $f(\nu)$ that has a $\pi$-defect at $\nu$ and all other defects at non-Clifford spiders $\nu'$ where $\nu \preceq \nu'$.
\end{itemize}
We call $\ell_X(i)$ and $\ell_Z(i)$ the \textit{logical semiwebs} of an input $i$, and we call $f(\nu)$ the \textit{flow semiweb} of a non-Clifford spider $\nu$.
\end{definition}

\begin{example}\label{ex:ZX-flow}
  The following ZX-diagram has the following logical Z- and X-semiwebs for input $i_1$ and flow semiwebs for the two non-Clifford spiders $\nu_1, \nu_2$:
  
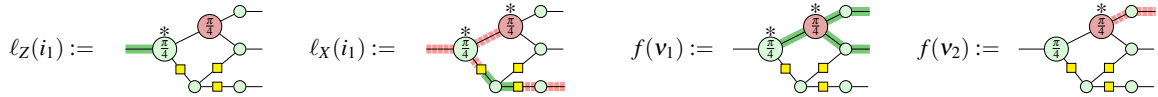
\begin{equation*}
    \begin{tikzpicture}
	\begin{pgfonlayer}{nodelayer}
		\node [style=X phase dot] (7) at (-2.75, 0.75) {$\frac\pi4$};
		\node [style=Z phase dot] (8) at (-4.25, 0) {$\frac\pi4$};
		\node [style=none] (9) at (-5.5, 0) {};
		\node [style=Z dot] (11) at (-1.75, 1.25) {};
		\node [style=Z dot] (12) at (-1.75, 0) {};
		\node [style=none] (13) at (-1, 1.25) {};
		\node [style=none] (14) at (-1, 0) {};
		\node [style=Z dot] (15) at (-3.25, -1.25) {};
		\node [style=hadamard] (16) at (-3.725, -0.675) {};
		\node [style=hadamard] (17) at (-2.5, -1.25) {};
		\node [style=Z dot] (18) at (-1.75, -1.25) {};
		\node [style=none] (19) at (-1, -1.25) {};
		\node [style=hadamard] (20) at (-2.5, -0.625) {};
		\node [style=X phase dot] (21) at (7.25, 0.75) {$\frac\pi4$};
		\node [style=Z phase dot] (22) at (5.75, 0) {$\frac\pi4$};
		\node [style=none] (23) at (4.5, 0) {};
		\node [style=Z dot] (24) at (8.25, 1.25) {};
		\node [style=Z dot] (25) at (8.25, 0) {};
		\node [style=none] (26) at (9, 1.25) {};
		\node [style=none] (27) at (9, 0) {};
		\node [style=Z dot] (28) at (6.75, -1.25) {};
		\node [style=hadamard] (29) at (6.275, -0.675) {};
		\node [style=hadamard] (30) at (7.5, -1.25) {};
		\node [style=Z dot] (31) at (8.25, -1.25) {};
		\node [style=none] (32) at (9, -1.25) {};
		\node [style=hadamard] (33) at (7.5, -0.625) {};
		\node [style=none] (34) at (-2.75, 1.25) {*};
		\node [style=none] (35) at (5.75, 0.5) {*};
		\node [style=none] (36) at (7.25, 1.25) {*};
		\node [style=none] (37) at (-8, 0) {$\ell_Z(i_1) :=$};
		\node [style=none] (67) at (2, 0) {$\ell_X(i_1) :=$};
	\end{pgfonlayer}
	\begin{pgfonlayer}{edgelayer}
		\draw [style=Z Web] (9.center) to (8);
		\draw [style=Z Web] (7) to (12);
		\draw [style=Z Web] (7) to (11);
		\draw [style=Z Web] (7) to (8);
		\draw [style=Z Web] (11) to (13.center);
		\draw [style=Z Web] (12) to (14.center);
		\draw (8) to (16);
		\draw (16) to (15);
		\draw (19.center) to (18);
		\draw (18) to (17);
		\draw (17) to (15);
		\draw (15) to (20);
		\draw (20) to (12);
		\draw [style=X Web] (23.center) to (22);
		\draw (21) to (25);
		\draw (21) to (24);
		\draw [style=X Web] (21) to (22);
		\draw (24) to (26.center);
		\draw (25) to (27.center);
		\draw [style=X Web] (22) to (29);
		\draw [style=Z Web] (29) to (28);
		\draw [style=X Web] (32.center) to (31);
		\draw [style=X Web] (31) to (30);
		\draw [style=Z Web] (30) to (28);
		\draw (28) to (33);
		\draw (33) to (25);
	\end{pgfonlayer}
\end{tikzpicture} \qquad \begin{tikzpicture}
	\begin{pgfonlayer}{nodelayer}
		\node [style=X phase dot] (0) at (-2.75, 0.75) {$\frac\pi4$};
		\node [style=Z phase dot] (1) at (-4.25, 0) {$\frac\pi4$};
		\node [style=none] (2) at (-5.5, 0) {};
		\node [style=Z dot] (3) at (-1.75, 1.25) {};
		\node [style=Z dot] (4) at (-1.75, 0) {};
		\node [style=none] (5) at (-1, 1.25) {};
		\node [style=none] (6) at (-1, 0) {};
		\node [style=Z dot] (7) at (-3.25, -1.25) {};
		\node [style=hadamard] (8) at (-3.75, -0.65) {};
		\node [style=hadamard] (9) at (-2.5, -1.25) {};
		\node [style=Z dot] (10) at (-1.75, -1.25) {};
		\node [style=none] (11) at (-1, -1.25) {};
		\node [style=hadamard] (12) at (-2.5, -0.625) {};
		\node [style=X phase dot] (13) at (6.75, 0.75) {$\frac\pi4$};
		\node [style=Z phase dot] (14) at (5.25, 0) {$\frac\pi4$};
		\node [style=none] (15) at (4, 0) {};
		\node [style=Z dot] (16) at (7.75, 1.25) {};
		\node [style=Z dot] (17) at (7.75, 0) {};
		\node [style=none] (18) at (8.5, 1.25) {};
		\node [style=none] (19) at (8.5, 0) {};
		\node [style=Z dot] (20) at (6.25, -1.25) {};
		\node [style=hadamard] (21) at (5.75, -0.625) {};
		\node [style=hadamard] (22) at (7, -1.25) {};
		\node [style=Z dot] (23) at (7.75, -1.25) {};
		\node [style=none] (24) at (8.5, -1.25) {};
		\node [style=hadamard] (25) at (7, -0.625) {};
		\node [style=none] (26) at (-4.25, 0.5) {*};
		\node [style=none] (28) at (6.75, 1.25) {*};
		\node [style=none] (29) at (-2.75, 1.25) {*};
		\node [style=none] (30) at (-7.5, 0) {$f(\nu_1) :=$};
		\node [style=none] (31) at (2, 0) {$f(\nu_2) :=$};
	\end{pgfonlayer}
	\begin{pgfonlayer}{edgelayer}
		\draw (2.center) to (1);
		\draw [style=Z Web] (0) to (4);
		\draw [style=Z Web] (0) to (3);
		\draw [style=Z Web] (0) to (1);
		\draw [style=Z Web] (3) to (5.center);
		\draw [style=Z Web] (4) to (6.center);
		\draw (11.center) to (10);
		\draw (10) to (9);
		\draw (9) to (7);
		\draw (7) to (12);
		\draw (12) to (4);
		\draw (15.center) to (14);
		\draw (13) to (17);
		\draw [style=X Web] (13) to (16);
		\draw (13) to (14);
		\draw [style=X Web] (16) to (18.center);
		\draw (17) to (19.center);
		\draw (24.center) to (23);
		\draw (23) to (22);
		\draw (22) to (20);
		\draw (20) to (25);
		\draw (25) to (17);
		\draw (1) to (7);
		\draw (14) to (20);
	\end{pgfonlayer}
\end{tikzpicture}
  \end{equation*}
  Here we have marked the spiders with a defect with $*$. This defines a ZX-flow with $\nu_1 \preceq \nu_2$.
\end{example}

One way we can interpret a ZX-flow for a diagram $D$ is that it gives us a recipe for implementing $D$ using a Clifford isometry $D'$ (thought of as a fixed resource, generalising e.g. graph states in the one-way model) and adaptive, single-qubit, non-Clifford measurements. A spider $\nu$ with non-Clifford phase $\alpha$ represents a single-qubit measurement in the basis $\{ \begin{tikzpicture}
	\begin{pgfonlayer}{nodelayer}
		\node [style=none] (0) at (-1, 0) {};
		\node [style=Z phase dot] (1) at (0.5, 0) {$\ \alpha\!+\!k\pi\ $};
	\end{pgfonlayer}
	\begin{pgfonlayer}{edgelayer}
		\draw (0.center) to (1);
	\end{pgfonlayer}
\end{tikzpicture}
 \}_{k = 0, 1}$. If $k = 0$, we got the desired outcome, so we proceed to the next measurement. If $k = 1$, we have an undesired $\pi$ phase, which we can cancel with another $\pi$-defect by ``firing'' the flow semiweb $f(\nu)$, as we did for Pauli webs in Example~\ref{ex:pauli-web-fire}. However, unlike firing Pauli webs, firing semiwebs can change the angles of some spiders. Notably, it will send $\alpha + \pi$ to $\alpha$, as desired, but it may also change some future measurement angles, located at the defects of $f(\nu)$.

\begin{example}
  Taking the diagram with ZX-flow from Example~\ref{ex:ZX-flow}, the flow semiweb $f(\nu_1)$ has a $\pi$ defect at $\nu_1$ and a $-\frac{\pi}{2}$ defect at $\nu_2$. This means if we get an unwanted outcome at $\nu_1$, we can push it into the future by firing $f(\nu_1)$:
  
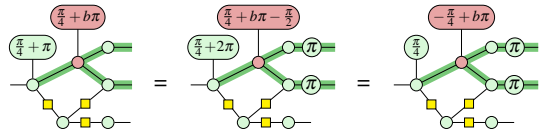
\begin{equation*}
    \begin{tikzpicture}
	\begin{pgfonlayer}{nodelayer}
		\node [style=X dot] (0) at (-2.75, 0.75) {};
		\node [style=Z dot] (1) at (-4.25, 0) {};
		\node [style=none] (2) at (-5, 0) {};
		\node [style=Z dot] (3) at (-1.75, 1.25) {};
		\node [style=Z dot] (4) at (-1.75, 0) {};
		\node [style=none] (5) at (-1, 1.25) {};
		\node [style=none] (6) at (-1, 0) {};
		\node [style=Z dot] (7) at (-3.25, -1.25) {};
		\node [style=hadamard] (8) at (-3.75, -0.65) {};
		\node [style=hadamard] (9) at (-2.5, -1.25) {};
		\node [style=Z dot] (10) at (-1.75, -1.25) {};
		\node [style=none] (11) at (-1, -1.25) {};
		\node [style=hadamard] (12) at (-2.5, -0.625) {};
		\node [style=Z phase dot] (13) at (-4.25, 1.25) {$\ \frac\pi4+\pi\ $};
		\node [style=X phase dot] (14) at (-2.75, 2.25) {$\ \frac\pi4+b\pi\ $};
		\node [style=X dot] (51) at (3.25, 0.75) {};
		\node [style=Z dot] (52) at (1.75, 0) {};
		\node [style=none] (53) at (1, 0) {};
		\node [style=Z dot] (54) at (4.25, 1.25) {};
		\node [style=Z dot] (55) at (4.25, 0) {};
		\node [style=none] (56) at (5.75, 1.25) {};
		\node [style=none] (57) at (5.75, 0) {};
		\node [style=Z dot] (58) at (2.75, -1.25) {};
		\node [style=hadamard] (59) at (2.25, -0.65) {};
		\node [style=hadamard] (60) at (3.5, -1.25) {};
		\node [style=Z dot] (61) at (4.25, -1.25) {};
		\node [style=none] (62) at (5, -1.25) {};
		\node [style=hadamard] (63) at (3.5, -0.625) {};
		\node [style=Z phase dot] (64) at (1.75, 1.25) {$\ \frac\pi4\!+\!2\pi\ $};
		\node [style=X phase dot] (65) at (3.25, 2.25) {$\ \frac\pi4+b\pi-\frac\pi2\ $};
		\node [style=none] (66) at (0, 0) {=};
		\node [style=Z tiny dot] (67) at (5.024999999999999, 1.25) {$\pi$};
		\node [style=Z tiny dot] (68) at (5, 0) {$\pi$};
		\node [style=X dot] (69) at (10, 0.75) {};
		\node [style=Z dot] (70) at (8.5, 0) {};
		\node [style=none] (71) at (7.75, 0) {};
		\node [style=Z dot] (72) at (11, 1.25) {};
		\node [style=Z dot] (73) at (11, 0) {};
		\node [style=none] (74) at (12.5, 1.25) {};
		\node [style=none] (75) at (12.5, 0) {};
		\node [style=Z dot] (76) at (9.5, -1.25) {};
		\node [style=hadamard] (77) at (9, -0.6499999999999999) {};
		\node [style=hadamard] (78) at (10.25, -1.25) {};
		\node [style=Z dot] (79) at (11, -1.25) {};
		\node [style=none] (80) at (11.75, -1.25) {};
		\node [style=hadamard] (81) at (10.25, -0.625) {};
		\node [style=Z phase dot] (82) at (8.5, 1.25) {$\ \frac\pi4\ $};
		\node [style=X phase dot] (83) at (10, 2.25) {$\ -\frac\pi4+b\pi\ $};
		\node [style=Z tiny dot] (84) at (11.774999999999999, 1.25) {$\pi$};
		\node [style=Z tiny dot] (85) at (11.75, 0) {$\pi$};
		\node [style=none] (86) at (6.75, 0) {=};
	\end{pgfonlayer}
	\begin{pgfonlayer}{edgelayer}
		\draw (2.center) to (1);
		\draw [style=Z Web] (0) to (4);
		\draw [style=Z Web] (0) to (3);
		\draw [style=Z Web] (0) to (1);
		\draw [style=Z Web] (3) to (5.center);
		\draw [style=Z Web] (4) to (6.center);
		\draw (11.center) to (10);
		\draw (10) to (9);
		\draw (9) to (7);
		\draw (7) to (12);
		\draw (12) to (4);
		\draw (14) to (0);
		\draw (13) to (1);
		\draw (53.center) to (52);
		\draw [style=Z Web] (51) to (55);
		\draw [style=Z Web] (51) to (54);
		\draw [style=Z Web] (51) to (52);
		\draw [style=Z Web] (54) to (56.center);
		\draw [style=Z Web] (55) to (57.center);
		\draw (62.center) to (61);
		\draw (61) to (60);
		\draw (60) to (58);
		\draw (65) to (51);
		\draw (64) to (52);
		\draw (52) to (58);
		\draw (58) to (55);
		\draw (1) to (7);
		\draw (71.center) to (70);
		\draw [style=Z Web] (69) to (73);
		\draw [style=Z Web] (69) to (72);
		\draw [style=Z Web] (69) to (70);
		\draw [style=Z Web] (72) to (74.center);
		\draw [style=Z Web] (73) to (75.center);
		\draw (80.center) to (79);
		\draw (79) to (78);
		\draw (78) to (76);
		\draw (83) to (69);
		\draw (82) to (70);
		\draw (70) to (76);
		\draw (76) to (73);
	\end{pgfonlayer}
\end{tikzpicture}
  \end{equation*}
  Note that the defect at the $\frac\pi4$ X-spider results in flipping the $\frac\pi4$ to $-\frac\pi4$, meaning that its measurement angle depends on the outcome of measuring the previous $\frac\pi4$ spider.
\end{example}

Whereas ZX-flow assigns a flow semiweb to each non-Clifford spider, we can define a stricter notion that requires a flow semiweb for every spider.

\begin{definition}
  A \textit{strong ZX-flow} for a diagram $D$ consists of a partial ordering $\preceq \subseteq S \times S$ on all spiders $S \subseteq N_D$ and the following Pauli semiwebs:
\begin{itemize}
  \item for every input wire $i \in I$, a pair of \textit{logical semiwebs} $\ell_Z(i)$ and $\ell_X(i)$ whose restrictions to input wires are $Z_i$ and $X_i$, respectively; and
  \item for every spider $\nu \in S$, a \textit{flow semiweb} $f(\nu)$ that has a $\pi$-defect at $\nu$ and all other defects at spiders $\nu'$ where $\nu \preceq \nu'$.
\end{itemize}
\end{definition}

As we will see in Section~\ref{sec:pauli-flow-equiv}, this stronger notion of ZX-flow is equivalent to Pauli flow in a strict sense. Namely, graph-like ZX-diagrams have strong ZX-flow if and only if they have Pauli flow. However, this stronger notion is not preserved by arbitrary Clifford rewrite rules, hence it is more easily broken by the kinds of simplifications we would like to do on ZX-diagrams.

It is not immediately clear that the existence of a strong ZX-flow implies the existence of a ZX-flow, because the former allows Pauli semiwebs to have defects at arbitrary spiders, not just the non-Clifford spiders. To solve this issue, we will use the technique of \textit{focusing}. In generalised flow and Pauli flow, focusing the flow corresponds to delaying measurement corrections in MBQC as long as possible. For ZX-flow, we see this concept reflected in trying to make semiwebs satisfy the \textbf{parity} condition from Theorem~\ref{thm:pauli-web-conds} for as many spiders $\nu$ as possible. For a non-Clifford $Z$-spider (resp. $X$-spider), this means $\nu$ should have even $Z$-support (resp. $X$-support) in its neighborhood, but it might still be a defect. For Clifford spiders, this means all the Pauli web criteria are satisfied, so there will not be a defect at $\nu$.

The following definition makes sense both for ZX-flow and strong ZX-flow, so we will state it as a single definition and corresponding focusing result.

\begin{definition}\label{def:focused-zx-flow}
  A \textit{focused (strong) ZX-flow} is a (strong) ZX-flow whose logical semiwebs satisfy the \textbf{parity} condition at every spider and whose flow semiwebs $f(\nu)$ satisfy the \textbf{parity} condition at every spider $\nu' \neq \nu$.
\end{definition}

\begin{theorem}\label{thm:focused-zx-flow}
  ZX-diagrams admit (strong) ZX-flow if and only if they admit \emph{focused} (strong) ZX-flow.
\end{theorem}

\begin{proof}[Proof of Theorem~\ref{thm:focused-zx-flow}]
  One direction is trivial. For the other direction, we prove this constructively by efficiently transforming a (strong) ZX-flow $(\ell_Z, \ell_X, f, \preceq)$ into a focused (strong) ZX-flow. The technique is the same in both cases, where the only difference is whether we operate on just the non-Clifford spiders in case of normal ZX-flow and all spiders for strong ZX-flow.
  
  We pick the $\preceq$-maximal spider $\nu \in T$ (or $S$ in the case of strong ZX-flow) such that $f(\nu)$ violates the \textbf{parity} condition for at least one spider $\nu' \neq \nu$. Let $\nu_1, \ldots, \nu_k$ be the spiders where the \textbf{parity} condition is violated. Since these must all be greater than $\nu$ with respect to $\preceq$, it must be the case that $f(\nu_j)$ can only violate the \textbf{parity} condition at $\nu_j$ (because otherwise $\nu$ would not be maximal in having this property). Note that $f(\nu_j)$ has a $\pi$-defect at $\nu_j$ and that this implies the \textbf{parity} condition is violated at $\nu_j$. The product of two semiwebs that violate the \textbf{parity} condition at a spider $\nu'$ will no longer violate it. Hence, the product $\vec w = f(\nu)f(\nu_1)\ldots f(\nu_k)$ will violate the \textbf{parity} condition at $\nu$ and nowhere else. We therefore set $f(\nu) := \vec w$ and then repeat until none of the flow semiwebs violate the focusing condition. Once this is done, we can remove any \textbf{parity} violations from the logical semiwebs $\ell_Z(i), \ell_X(i)$ similarly, by multiplying with the flow semiwebs of the spiders where there are \textbf{parity} violations.
\end{proof}

Being focused implies that none of the logical semiwebs have defects at Clifford nodes and the flow semiwebs $f(\nu)$ only have a defect at a Clifford node if $\nu$ is itself Clifford. Hence, a strong focused ZX-flow on all spiders will restrict to a valid ZX-flow on the non-Clifford spiders. The following is then an immediate consequence of Theorem~\ref{thm:focused-zx-flow}.

\begin{corollary}\label{cor:strong-implies-weak}
If a ZX-diagram admits a strong ZX-flow, it admits a ZX-flow.
\end{corollary}

\subsection{Clifford preservation of ZX-flow}\label{sec:flow-preservation}

\begin{figure}
\centering
\input{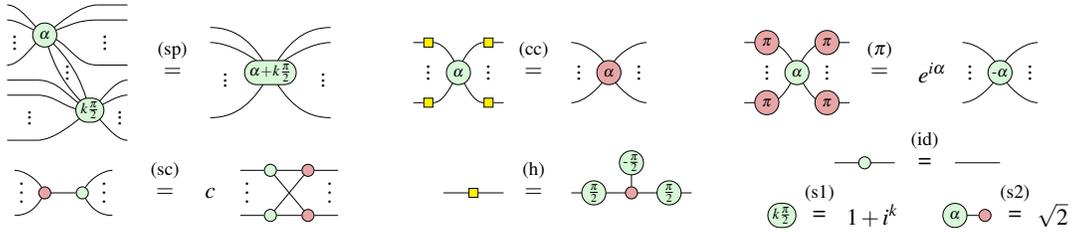}
\caption{The extended Clifford ZX-calculus, where $\alpha \in \mathbb R$, $k \in \mathbb Z$, and $c = \sqrt{2}^{(n-1)(m-1)}$ for $n$ inputs and $m$ outputs. The rules are pronounced: \sprule spider, \ccrule colour-change, \pirule $\pi$ stabiliser, \scrule strong complementarity, \hadrule Hadamard, \idrule identity, \srule and \stworule scalar rules.}
\label{fig:Clifford-rules}
\end{figure}

In this section, we will show that all of the rules of the Clifford ZX-calculus preserve ZX-flow. In fact, we can show this fact for a bit more than the usual Clifford ZX-calculus, as we can allow rules that have arbitrary phase parameters in certain locations; see Figure~\ref{fig:Clifford-rules}. In particular, spider fusion with arbitrary angles preserves ZX-flow in either direction with only one exception: it cannot be applied to unfuse a spider with a Clifford angle into two spiders with non-Clifford angles.

The key observation is that when we replace a diagram consisting of only Clifford spiders with another one, we can always locally update Pauli semiwebs thanks to Theorem~\ref{thm:pauli-web-stabilisers}. For example:
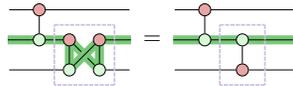
\ctikzfig{local-update}
We formalise the notion of local preservation of Pauli webs by purely Clifford rules in the following Lemma.

\begin{lemma}\label{lem:clifford-zx-pres}
For two Clifford ZX-diagrams $D$, $D'$, if $\llbracket D \rrbracket = \llbracket D' \rrbracket$, then for any Pauli web $\vec w$ on $D$ we can find a corresponding Pauli web $\vec w'$ on $D'$ such that $\vec w$ and $\vec w'$ have the same support on the boundary.
\end{lemma}

It then only remains to show that any non-Clifford spiders that we modify can be assigned valid flow semiwebs after a rule application. This gives rise to the following result.

\begin{theorem}\label{thm:ext-clifford-zx-pres}
  The rules of the extended Clifford ZX-calculus (Figure~\ref{fig:Clifford-rules}) preserve ZX-flow.
\end{theorem}

\begin{proof}
  All the semiwebs used in the ZX-flow only have defects at non-Clifford spiders, so that in the neighbourhood of the Clifford spiders, the semiwebs satisfy all the local Pauli web rules.
  Preservation of ZX-flow for ZX-rules involving only Clifford spiders then follows from Lemma~\ref{lem:clifford-zx-pres}. It then remains to consider \sprule, \ccrule, \pirule, and \stworule for $\alpha \neq j\frac{\pi}{2}$.

  For \sprule, let $\nu$ be the non-Clifford spider on the LHS and $\nu'$ be the non-Clifford spider on the RHS. If we apply \sprule left-to-right, we define $\preceq$ on the resulting diagram by replacing $\nu$ with $\nu'$. We then need to show $\nu'$ has a suitable flow semiweb, that the other semiwebs in the diagram remain valid, and that the defects respect the new ordering $\preceq$.

  The wires on the RHS are a subset of the wires on the LHS, so we can restrict all logical and flow webs $\vec w$ to the wires of the new diagram to obtain a new semiweb $\vec w^-$. All $\vec w^-$ will satisfy the \textbf{H} condition, and if any of the input/output wires on the LHS of \sprule have $X$-support for any semiweb, then they all must have $X$-support, so all $\vec w^-$ also satisfy the \textbf{all-or-nothing} condition, so that all the previously existing semiwebs remain valid. 

  Now, let $f(\nu') := f(\nu)^-$. We know that $f(\nu)$ had a $\pi$-defect at $\nu$, so it must be the case that oddly many wires adjacent to $\nu$ have $Z$-support and no adjacent wires have $X$-support.
  But then the Clifford spider on the LHS must have $Z$-support on an even number of adjacent wires.
  We then distinguish two cases. Either an odd number of internal wires have $Z$-support, or an even number do. If this is odd, then the non-Clifford spider has an even number of boundary wires with $Z$-support, and the Clifford spider has an odd amount. After fusion this hence combines as even+odd = odd boundary wires with $Z$-support. If instead an even number of internal wires has $Z$-support, then the non-Clifford spider has an odd number of boundary wires with $Z$-support and the Clifford spider has an odd amount, so that again we have even+odd = odd boundary wires with $Z$-support after fusion.
  Hence in both cases $f(\nu')$ must have $Z$-support on oddly many adjacent wires and $X$-support on no adjacent wires, so that $f(\nu')$ satisfies the flow condition.

  Finally, for any other non-Clifford spider $\nu''$, if $f(\nu'')^-$ has a defect at $\nu'$, then it previously had a defect at $\nu$, so $\nu'' \preceq \nu'$, so that the defects still respect $\preceq$.

  Applying \sprule from right-to-left is similar, except we extend each logical and flow semiweb $\vec w$ to a new semiweb $\vec w^+$ which may have additional support on the newly-created wire(s) between the unfused spiders. We do this such that the Pauli web condition is satisfied for the Clifford spider on the LHS. Namely, the new wires have $X$-support if and only if the boundary wires do, and we introduce $Z$-support on 0 or 1 wires in order to make the $Z$-support of the Clifford spider even. Intuitively, we can think of whether we colour the connecting edge as capturing whether the Pauli semiweb ``exits'' through the Clifford spider. For example:
  \ctikzfig{spider-unfusion-example}
  This will imply that the $Z$-support of $f(\nu')^+$ is odd, so we can set $f(\nu) := f(\nu')^+$.

  The arguments for \ccrule and \pirule are very similar, as the LHS and RHS both contain exactly 1 non-Clifford spider. In fact, these cases are a bit simpler, as the extended Pauli semiwebs when applying the rules from right-to-left are uniquely fixed. This is because the $H$-gates in the \ccrule and the $\pi$-spiders in the \pirule uniquely fix the colour of the internal wires based on the colour of the inputs/outputs. For example:
  \ctikzfig{colour-change-example}

  Applying \stworule left-to-right trivially preserves ZX-flow. From right-to-left, we can give the non-Clifford spider a flow semiweb just by introducing $Z$-support on the one wire in the LHS.
\end{proof}

ZX-flow is also preserved by spider (un)fusions involving non-Clifford angles, as long as we don't use it to create non-Clifford spiders at a location where there was previously only a Clifford spider.

\begin{theorem}\label{thm:full-sp-fusion-pres}
Spider fusion between a pair of non-Clifford spiders preserves ZX-flow from left-to-right, and it preserves ZX-flow from right-to-left as long as we don't unfuse a Clifford spider into a pair of non-Clifford spiders.
\end{theorem}

\begin{proof}
When we apply the rule from left-to-right, the only case not covered by Theorem~\ref{thm:ext-clifford-zx-pres} is when $\alpha$ and $\beta$ are both non-Clifford. As in the proof of that theorem, we define new Pauli semiwebs by restricting the old ones. If $\alpha + \beta$ is Clifford, we are done. If $\alpha + \beta$ is non-Clifford, the spider on the RHS $\nu'$ inherits the flow semiweb and its position w.r.t. $\preceq$ from either spider $\nu$ on the LHS.

Now let $\nu''$ be any other spider in the diagram we apply the rewrite to and suppose it has a defect at one of the spiders in the LHS. We can assume without loss of generality that the ZX-flow is focused, hence this defect must not be a $\pi$-defect. Hence, $f(\nu'')$ must have $X$-support on all the edges in the LHS of \fullsprule, so it must have a defect at both spiders, so $\nu''$ is before both spiders. This implies that if $f(\nu'')^-$ has a defect at the spider $\nu'$ on the RHS of \fullsprule, then $\nu'' \preceq \nu'$.

Applying the rule from right-to-left, again we only need to consider the case where both spiders on the LHS are non-Clifford. We place them both at the same location as the spider on the RHS w.r.t. $\preceq$. The ordering between the two spiders on the LHS is irrelevant. We extend the Pauli semiwebs as described in the proof of Theorem~\ref{thm:ext-clifford-zx-pres} and note that any spider $\nu''$ with a defect at one of the two spiders on the LHS must be in the past, as it would have also had a defect on the spider in the RHS.
\end{proof}

\subsection{Equivalence of ZX-Flow and Pauli flow}\label{sec:pauli-flow-equiv}

Once we have a strong ZX-flow, we can define a Pauli flow by setting correction sets $c(\nu) := B$, where $B$ is a generating set of spiders for $f(\nu)$ (cf.~Theorem~\ref{thm:generating-set-of-spiders}). Conversely, correction sets give flow semiwebs by taking products of basic semiwebs $b_\nu$ for $\nu \in c(\nu)$. This gives the following equivalence, which we prove in detail in Appendix~\ref{app:pauli-flow-equiv}.

\begin{lemma}\label{lem:strong-zxflow-iff-pauli-flow}
A graph-like ZX-diagram has strong ZX-flow if and only if its induced labelled open graph has Pauli flow.
\end{lemma}

A consequence of this lemma and Theorem~\ref{thm:ext-clifford-zx-pres} is that any ZX-diagram that is Clifford-equivalent to one with Pauli flow has ZX-flow. For the converse, if we start with a ZX-diagram that has ZX-flow, we can use Clifford rewrites to make it graph-like and remove all unnecessary Clifford spiders. Once we do this, we can show that the ZX-flow is actually now a strong ZX-flow, hence we can apply the lemma above to get a Pauli flow. Using this technique, we prove one of our main results.

\begin{theorem}\label{thm:zxflow-iff-pauli-flow}
  A ZX-diagram $D$ has ZX-flow if and only if it is Clifford-equivalent to a graph-like ZX-diagram $D'$ with Pauli flow.
\end{theorem}

\begin{proof}
  Suppose $D$ is Clifford-equivalent to some graph-like $D'$ with Pauli flow. Then, $D'$ has strong ZX-flow by Lemma~\ref{lem:strong-zxflow-iff-pauli-flow}. By Corollary~\ref{cor:strong-implies-weak} $D'$ then also has regular ZX-flow. Then, by Theorem~\ref{thm:ext-clifford-zx-pres}, $D$ must have ZX-flow.

  Conversely, assume $D$ has ZX-flow. Using the Clifford simplification routine from e.g.~\cite{kissinger2019tcount}, we can apply Clifford rewrites to reduce $D$ to its ``skeleton'' $D'$, i.e. a graph-like diagram whose only interior Clifford spiders are part of \emph{phase gadgets}, and which might have Hadamards on the boundary wires. That is, every interior Clifford spider $\nu$ must be connected via an $H$-edge to a 1-legged non-Clifford spider $\nu'$. $D'$ must have a ZX-flow by Theorem~\ref{thm:ext-clifford-zx-pres}, but we can furthermore give it a strong ZX-flow by placing every interior Clifford spider $\nu$ directly before its associated 1-legged non-Clifford spider $\nu'$ in $\preceq$. We then let $f(\nu) := \vec b_{\nu'}$, which only covers the H-edge between $\nu'$ and $\nu$. This basic semiweb has a $\pi$ defect at $\nu$ and another defect at $\nu'$, so it is a valid flow semiweb for $\nu$.

  Place input Clifford spiders as minimal elements in $\preceq$. For a Clifford spider $\nu$ adjacent to an input wire $i$, we construct $f(\nu)$ by multiplying one of the two logical semiwebs $\ell_Z(i)$ or $\ell_X(i)$ with an edge semiweb that contains $i$ to remove its support on the input wire. Output Clifford spiders can be placed arbitrarily w.r.t. $\preceq$, and we let $f(\nu)$ be the edge semiweb that contains its adjacent output. 

  We have then shown that $D'$ has strong ZX-flow. In order to invoke Lemma~\ref{lem:strong-zxflow-iff-pauli-flow}, we must first make it fully graph-like, which means we need to introduce identity spiders on those boundary wires that contain Hadamards. Introducing these identities trivially preserves the desired properties of all the existing semiwebs, so we just have to find semiwebs for these identity spiders. Since these are all boundary spiders, we do this in the same way as we did with the previous boundary spiders.

  We have thus shown that $D'$ has a strong ZX-flow and is graph-like, so by Lemma~\ref{lem:strong-zxflow-iff-pauli-flow}, it has Pauli flow.
\end{proof}

\section{Circuit extraction}\label{sec:extraction}

We will now show how we can perform circuit extraction for a ZX-diagram with ZX-flow. The main idea is that we can interpret a ZX-diagram with a focused ZX-flow directly as a Clifford circuit, followed by a sequence of \textit{Pauli exponentials} (a.k.a. Pauli gadgets or Pauli-product rotations).

\begin{definition}
For a Pauli operator $\vec P = P_1 \otimes \ldots \otimes P_n$ and angle $\alpha \in \mathbb R$, a \textit{Pauli exponential} is an $n$-qubit unitary of the form $\vec P[\alpha] := e^{-\frac{i \alpha}{2} \vec P}$, or diagrammatically:
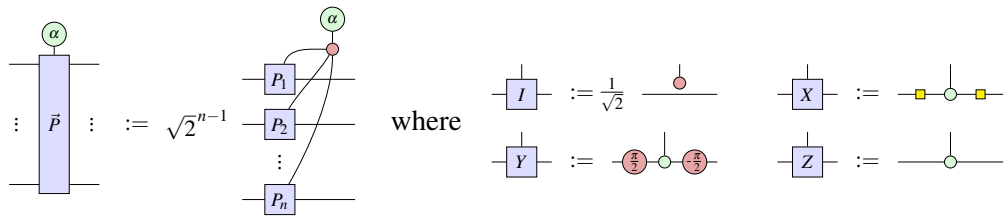
\[
\begin{tikzpicture}
	\begin{pgfonlayer}{nodelayer}
		\node [style=paulibox] (20) at (2.5, 1.5) {$P_1$};
		\node [style=none] (21) at (1.25, 1.5) {};
		\node [style=none] (22) at (5, 1.5) {};
		\node [style=paulibox] (24) at (2.5, -2.5) {$P_n$};
		\node [style=none] (25) at (1.25, -2.5) {};
		\node [style=none] (26) at (5, -2.5) {};
		\node [style=none, rotate=90] (27) at (2.5, -1.25) {...};
		\node [style=X dot] (30) at (4.25, 2.5) {};
		\node [style=paulibox] (31) at (2.5, 0) {$P_2$};
		\node [style=none] (32) at (1.25, 0) {};
		\node [style=none] (33) at (5, 0) {};
		\node [style=paulibox, minimum height=2.3cm] (34) at (-5, 0) {$\vec P$};
		\node [style=none] (35) at (-6.5, 2) {};
		\node [style=none] (36) at (-3.5, 2) {};
		\node [style=none] (37) at (-6.5, -2) {};
		\node [style=none] (38) at (-3.5, -2) {};
		\node [style=none, rotate=90] (39) at (-6.25, 0) {...};
		\node [style=none, rotate=90] (40) at (-3.75, 0) {...};
		\node [style=none] (41) at (-2.25, 0) {$:=$};
		\node [style=none] (42) at (-0.25, 0) {$\sqrt{2}^{\,n-1}$};
		\node [style=Z phase dot] (43) at (-5, 3) {$\alpha$};
		\node [style=Z phase dot] (44) at (4.25, 3.5) {$\alpha$};
	\end{pgfonlayer}
	\begin{pgfonlayer}{edgelayer}
		\draw (21.center) to (22.center);
		\draw (25.center) to (26.center);
		\draw (32.center) to (33.center);
		\draw [in=-180, out=75] (20) to (30);
		\draw [in=240, out=60, looseness=0.75] (31) to (30);
		\draw [in=-90, out=60, looseness=0.75] (24) to (30);
		\draw (35.center) to (36.center);
		\draw (37.center) to (38.center);
		\draw (34) to (43);
		\draw (30) to (44);
	\end{pgfonlayer}
\end{tikzpicture}
\quad \textrm{where} \quad
\begin{tikzpicture}
	\begin{pgfonlayer}{nodelayer}
		\node [style=paulibox] (0) at (-7.25, 1) {$I$};
		\node [style=none] (1) at (-8.25, 1) {};
		\node [style=none] (2) at (-6.25, 1) {};
		\node [style=none] (3) at (-7.25, 2) {};
		\node [style=paulibox] (4) at (2.25, 1) {$X$};
		\node [style=none] (5) at (1.25, 1) {};
		\node [style=none] (6) at (3.25, 1) {};
		\node [style=none] (7) at (2.25, 2) {};
		\node [style=none] (8) at (-5.25, 1) {$:=$};
		\node [style=X dot] (9) at (-2, 1.375) {};
		\node [style=none] (10) at (-3.25, 1) {};
		\node [style=none] (11) at (-0.75, 1) {};
		\node [style=none] (12) at (-2, 2) {};
		\node [style=Z dot] (13) at (7, 1) {};
		\node [style=none] (14) at (5.25, 1) {};
		\node [style=none] (15) at (8.75, 1) {};
		\node [style=none] (16) at (7, 2) {};
		\node [style=hadamard] (17) at (6, 1) {};
		\node [style=hadamard] (18) at (8, 1) {};
		\node [style=none] (19) at (4.25, 1) {$:=$};
		\node [style=paulibox] (20) at (-7.25, -1.25) {$Y$};
		\node [style=none] (21) at (-8.25, -1.25) {};
		\node [style=none] (22) at (-6.25, -1.25) {};
		\node [style=none] (23) at (-7.25, -0.25) {};
		\node [style=Z dot] (24) at (-2.5, -1.25) {};
		\node [style=none] (25) at (-4.25, -1.25) {};
		\node [style=none] (26) at (-0.75, -1.25) {};
		\node [style=none] (27) at (-2.5, -0.25) {};
		\node [style=X phase dot] (28) at (-3.5, -1.25) {$\frac\pi2$};
		\node [style=none] (30) at (-5.25, -1.25) {$:=$};
		\node [style=X phase dot] (31) at (-1.5, -1.25) {-$\frac\pi2$};
		\node [style=paulibox] (32) at (2.25, -1.25) {$Z$};
		\node [style=none] (33) at (1.25, -1.25) {};
		\node [style=none] (34) at (3.25, -1.25) {};
		\node [style=none] (35) at (2.25, -0.25) {};
		\node [style=Z dot] (36) at (7, -1.25) {};
		\node [style=none] (37) at (5.25, -1.25) {};
		\node [style=none] (38) at (8.75, -1.25) {};
		\node [style=none] (39) at (7, -0.25) {};
		\node [style=none] (42) at (4.25, -1.25) {$:=$};
		\node [style=none] (43) at (-4.25, 1) {$\oneoversqrttwo$};
	\end{pgfonlayer}
	\begin{pgfonlayer}{edgelayer}
		\draw (1.center) to (2.center);
		\draw (0) to (3.center);
		\draw (5.center) to (6.center);
		\draw (4) to (7.center);
		\draw (10.center) to (11.center);
		\draw (9) to (12.center);
		\draw (14.center) to (15.center);
		\draw (13) to (16.center);
		\draw (21.center) to (22.center);
		\draw (20) to (23.center);
		\draw (25.center) to (26.center);
		\draw (24) to (27.center);
		\draw (33.center) to (34.center);
		\draw (32) to (35.center);
		\draw (37.center) to (38.center);
		\draw (36) to (39.center);
	\end{pgfonlayer}
\end{tikzpicture}
\]
\end{definition}

A useful feature of Pauli gadgets is that they commute through other maps in exactly the same way as Pauli operators. The following can be shown by direct calculation or using the Clifford ZX-calculus rules (see e.g.~\cite{KissingerWetering2024Book}).

\begin{proposition}\label{prop:push-Pauli}
  For a linear map $L : (\mathbb C^2)^{\otimes n} \to (\mathbb C^2)^{\otimes m}$ and Pauli operators $\vec P \in \mathcal P_n, \vec Q \in \mathcal P_m$, if $L \vec P = \vec Q L$, then $L \vec P[\alpha] = \vec Q[\alpha] L$.
\end{proposition}

Combining this with Theorem~\ref{thm:pauli-web-stabilisers}, we can see that Pauli webs enable us to push Pauli exponentials through a ZX-diagram. Suppose we take a ZX-diagram $D$ with focused ZX-flow and pick a non-Clifford spider that is maximal with respect to $\preceq$. If it is a $Z$-spider, we unfuse a 1-legged $Z$-spider $\nu$ with phase angle $\alpha$ connected via a plain edge. If it is an $X$-spider, we unfuse a 1-legged $Z$-spider connected via an $H$-edge.
In either case, we obtain a flow semiweb that has $Z$-support on the wire $w$ adjacent to $\nu$ and it has support on the outputs given by some Pauli $\vec P$. We will call the restriction $\vec P_O$ of a semiweb to the outputs of a diagram its \textit{output colouring}.

Since $\nu$ is maximal and the ZX-flow is focused, $f(\nu)$ only has a defect at $\nu$ itself, so we can treat $f(\nu)$ as a Pauli web on a sub-diagram $D'$ which contains everything in $D$ except for $\nu$. We then use this Pauli web to push the non-Clifford angle $\alpha$ out of the diagram as follows:
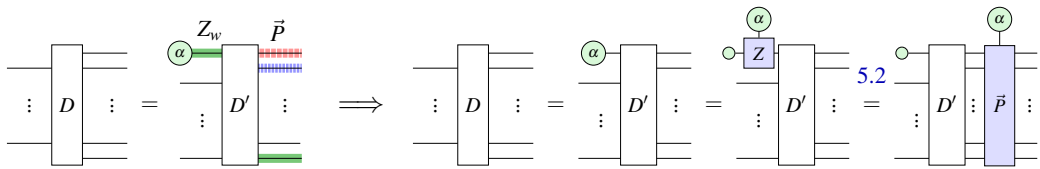
\[
\begin{tikzpicture}
	\begin{pgfonlayer}{nodelayer}
		\node [style=big gate, minimum height=2cm] (44) at (0, 0.125) {$D'$};
		\node [style=none] (47) at (2, 1.75) {};
		\node [style=none] (49) at (0.5, 1.75) {};
		\node [style=none] (50) at (2, -1.75) {};
		\node [style=none] (51) at (0.5, -1.75) {};
		\node [style=none, rotate=90] (52) at (1.5, 0) {...};
		\node [style=none] (53) at (2, 1.25) {};
		\node [style=none] (54) at (0.5, 1.25) {};
		\node [style=none] (55) at (2, -1.25) {};
		\node [style=none] (56) at (0.5, -1.25) {};
		\node [style=Z phase dot] (57) at (-2, 1.75) {$\alpha$};
		\node [style=none] (58) at (-0.5, 1.75) {};
		\node [style=none] (59) at (1.25, 2.5) {$\vec P$};
		\node [style=none] (60) at (-1, 2.5) {$Z_w$};
		\node [style=none] (61) at (-0.5, 0.75) {};
		\node [style=none] (62) at (-2, 0.75) {};
		\node [style=none] (63) at (-0.5, -1.75) {};
		\node [style=none] (64) at (-2, -1.75) {};
		\node [style=none, rotate=90] (65) at (-1.25, -0.5) {...};
		\node [style=big gate, minimum height=2cm] (66) at (-5.75, 0.125) {$D$};
		\node [style=none] (67) at (-3.75, 1.75) {};
		\node [style=none] (68) at (-5.25, 1.75) {};
		\node [style=none] (69) at (-3.75, -1.75) {};
		\node [style=none] (70) at (-5.25, -1.75) {};
		\node [style=none, rotate=90] (71) at (-4.25, 0) {...};
		\node [style=none] (72) at (-3.75, 1.25) {};
		\node [style=none] (73) at (-5.25, 1.25) {};
		\node [style=none] (74) at (-3.75, -1.25) {};
		\node [style=none] (75) at (-5.25, -1.25) {};
		\node [style=none] (80) at (-6.25, 1.25) {};
		\node [style=none] (81) at (-7.75, 1.25) {};
		\node [style=none] (82) at (-6.25, -1.25) {};
		\node [style=none] (83) at (-7.75, -1.25) {};
		\node [style=none, rotate=90] (84) at (-7, 0) {...};
		\node [style=none] (85) at (-3, 0) {$=$};
	\end{pgfonlayer}
	\begin{pgfonlayer}{edgelayer}
		\draw [style=X Web] (49.center) to (47.center);
		\draw [style=Z Web] (51.center) to (50.center);
		\draw [style=Y Web] (54.center) to (53.center);
		\draw (56.center) to (55.center);
		\draw [style=Z Web] (57) to (58.center);
		\draw (62.center) to (61.center);
		\draw (64.center) to (63.center);
		\draw (68.center) to (67.center);
		\draw (70.center) to (69.center);
		\draw (73.center) to (72.center);
		\draw (75.center) to (74.center);
		\draw (81.center) to (80.center);
		\draw (83.center) to (82.center);
	\end{pgfonlayer}
\end{tikzpicture}
\ \ \implies \ \ 
\begin{tikzpicture}
	\begin{pgfonlayer}{nodelayer}
		\node [style=paulibox, minimum height=2cm] (34) at (6.25, 0) {$\vec P$};
		\node [style=none, rotate=90] (39) at (5.375, 0) {...};
		\node [style=none, rotate=90] (40) at (7.25, 0) {...};
		\node [style=Z phase dot] (43) at (6.25, 2.875) {$\alpha$};
		\node [style=big gate, minimum height=2cm] (44) at (-5.75, 0.125) {$D'$};
		\node [style=none] (47) at (-4, 1.75) {};
		\node [style=none] (49) at (-5.25, 1.75) {};
		\node [style=none] (50) at (-4, -1.75) {};
		\node [style=none] (51) at (-5.25, -1.75) {};
		\node [style=none, rotate=90] (52) at (-4.5, 0) {...};
		\node [style=none] (53) at (-4, 1.25) {};
		\node [style=none] (54) at (-5.25, 1.25) {};
		\node [style=none] (55) at (-4, -1.25) {};
		\node [style=none] (56) at (-5.25, -1.25) {};
		\node [style=Z phase dot] (57) at (-7.25, 1.75) {$\alpha$};
		\node [style=none] (58) at (-6.25, 1.75) {};
		\node [style=none] (61) at (-6.25, 0.75) {};
		\node [style=none] (62) at (-7.75, 0.75) {};
		\node [style=none] (63) at (-6.25, -1.75) {};
		\node [style=none] (64) at (-7.75, -1.75) {};
		\node [style=none, rotate=90] (65) at (-7, -0.5) {...};
		\node [style=none] (66) at (-3.25, 0) {$=$};
		\node [style=big gate, minimum height=2cm] (67) at (-0.5, 0.125) {$D'$};
		\node [style=none] (68) at (1.25, 1.75) {};
		\node [style=none] (69) at (0, 1.75) {};
		\node [style=none] (70) at (1.25, -1.75) {};
		\node [style=none] (71) at (0, -1.75) {};
		\node [style=none, rotate=90] (72) at (0.75, 0) {...};
		\node [style=none] (73) at (1.25, 1.25) {};
		\node [style=none] (74) at (0, 1.25) {};
		\node [style=none] (75) at (1.25, -1.25) {};
		\node [style=none] (76) at (0, -1.25) {};
		\node [style=none] (78) at (-1, 1.75) {};
		\node [style=none] (79) at (-1, 0.75) {};
		\node [style=none] (80) at (-2.5, 0.75) {};
		\node [style=none] (81) at (-1, -1.75) {};
		\node [style=none] (82) at (-2.5, -1.75) {};
		\node [style=none, rotate=90] (83) at (-1.75, -0.5) {...};
		\node [style=Z dot] (84) at (-2.75, 1.75) {};
		\node [style=none] (85) at (2, 0) {$=$};
		\node [style=big gate, minimum height=2cm] (86) at (4.5, 0.125) {$D'$};
		\node [style=none] (87) at (7.5, 1.75) {};
		\node [style=none] (88) at (5, 1.75) {};
		\node [style=none] (89) at (7.5, -1.75) {};
		\node [style=none] (90) at (5, -1.75) {};
		\node [style=none] (92) at (7.5, 1.25) {};
		\node [style=none] (93) at (5, 1.25) {};
		\node [style=none] (94) at (7.5, -1.25) {};
		\node [style=none] (95) at (5, -1.25) {};
		\node [style=none] (97) at (4, 1.75) {};
		\node [style=none] (98) at (4, 0.75) {};
		\node [style=none] (99) at (2.75, 0.75) {};
		\node [style=none] (100) at (4, -1.75) {};
		\node [style=none] (101) at (2.75, -1.75) {};
		\node [style=none, rotate=90] (102) at (3.5, -0.5) {...};
		\node [style=Z dot] (103) at (3, 1.75) {};
		\node [style=big gate, minimum height=2cm] (104) at (-11.25, 0.125) {$D$};
		\node [style=none] (105) at (-9.25, 1.75) {};
		\node [style=none] (106) at (-10.75, 1.75) {};
		\node [style=none] (107) at (-9.25, -1.75) {};
		\node [style=none] (108) at (-10.75, -1.75) {};
		\node [style=none, rotate=90] (109) at (-9.75, 0) {...};
		\node [style=none] (110) at (-9.25, 1.25) {};
		\node [style=none] (111) at (-10.75, 1.25) {};
		\node [style=none] (112) at (-9.25, -1.25) {};
		\node [style=none] (113) at (-10.75, -1.25) {};
		\node [style=none] (114) at (-11.75, 1.25) {};
		\node [style=none] (115) at (-13.25, 1.25) {};
		\node [style=none] (116) at (-11.75, -1.25) {};
		\node [style=none] (117) at (-13.25, -1.25) {};
		\node [style=none, rotate=90] (118) at (-12.5, 0) {...};
		\node [style=none] (119) at (-8.5, 0) {$=$};
		\node [style=paulibox] (120) at (-1.75, 1.75) {$Z$};
		\node [style=Z phase dot] (121) at (-1.75, 2.875) {$\alpha$};
		\node [style=none] (122) at (2, 1) {\ref{prop:push-Pauli}};
	\end{pgfonlayer}
	\begin{pgfonlayer}{edgelayer}
		\draw (34) to (43);
		\draw (49.center) to (47.center);
		\draw (51.center) to (50.center);
		\draw (54.center) to (53.center);
		\draw (56.center) to (55.center);
		\draw (57) to (58.center);
		\draw (62.center) to (61.center);
		\draw (64.center) to (63.center);
		\draw (69.center) to (68.center);
		\draw (71.center) to (70.center);
		\draw (74.center) to (73.center);
		\draw (76.center) to (75.center);
		\draw (80.center) to (79.center);
		\draw (82.center) to (81.center);
		\draw (88.center) to (87.center);
		\draw (90.center) to (89.center);
		\draw (93.center) to (92.center);
		\draw (95.center) to (94.center);
		\draw (99.center) to (98.center);
		\draw (101.center) to (100.center);
		\draw (103) to (97.center);
		\draw (106.center) to (105.center);
		\draw (108.center) to (107.center);
		\draw (111.center) to (110.center);
		\draw (113.center) to (112.center);
		\draw (115.center) to (114.center);
		\draw (117.center) to (116.center);
		\draw (84) to (78.center);
		\draw (121) to (120);
	\end{pgfonlayer}
\end{tikzpicture}
\]

This strictly reduces the number of non-Clifford spiders in $D$, so at some point this will terminate when $D'$ is a Clifford ZX-diagram. Because the ZX-flow is focused, the logical semiwebs $\ell_Z(i), \ell_X(i)$ have no defects, so that they are logical Pauli webs. This will give us a stabiliser tableau for the initial Clifford part.

Using known procedures to synthesise Clifford isometries and Pauli exponentials, this gives us a complete procedure for extracting a ZX-diagram into a basic set of gates such as CNOT, $H$, $Z[\alpha]$.

Even though we showed this extraction procedure by rewriting the diagram, we can already read the form of the final extracted circuit from the focused ZX-flow data. This gives us the following theorem.

\begin{theorem}\label{thm:circuit-interp}
  Let $D$ be a ZX-diagram with focused ZX-flow. Fix an ordering of non-Clifford spiders $\nu_1, \ldots, \nu_m$ consistent with $\preceq$ such that the angle of $\nu_k$ is $\alpha_k$ and the output colouring of $f(\nu_k)$ is $\vec P_k$, then:
  \ctikzfig{extraction}
  where $C$ is a Clifford isometry with logical operators $\vec{\mathcal Z}_i, \vec{\mathcal X}_i$ given by the output colourings of $\ell_Z(i)$, respectively $\ell_X(i)$, on $D$.
\end{theorem}

\section{Conclusion}

We defined a new flow condition that is naturally suited to generic ZX-diagrams, which allows us to work with a broader class of diagrams while still retaining the ability to extract quantum computations, either via deterministic measurement patterns or quantum circuits. Whereas previous conditions were very closely related to graph states and the one-way model of measurement-based quantum computing, ZX-flow can be thought of as a ``multi-paradigm'' concept of runnability. Theorem~\ref{thm:circuit-interp}, where we give a direct interpretation of a focused ZX-flow as a unitary circuit in a particular form, gives evidence of this. It would be interesting to explore other interpretations of ZX-flow in other quantum computational paradigms such as Pauli-based computation and lattice surgery. It may be useful when thinking about these other paradigms to think about a ``Heisenberg-style'' formulation of ZX-flow, where logical and flow semiwebs have support backward in time rather than forward, reflecting changes to the relevant Pauli observables that might happen as a result of past measurements.

Another direction for future exploration is efficiently finding ZX-flow when it exists. Identifying the necessary Pauli webs amounts to solving systems of $\mathbb F_2$-linear equations, so one could imagine a na\"ive procedure where we pick non-Clifford spiders one at a time and try to assign them flow semiwebs using spiders we have already seen as defects. This procedure is likely to run in $O(n^4)$, but perhaps one could do better by incorporating ideas from the efficient Pauli flow finding algorithm of~\cite{mitosek2024algebraic}.

\textbf{Acknowledgements}:
AK is supported by the Engineering and Physical Sciences Research Council grant number EP/Z002230/1, \textit{(De)constructing quantum software (DeQS)}. JvdW acknowledges support by a Veni grant from the Dutch Research Council (NWO). An LLM was used to proofread this paper and check for minor notational errors.

\bibliographystyle{eptcs}
\bibliography{main}

\newpage

\appendix

\section{Correctness of the Definition of Pauli flow}\label{app:pauli-flow-def}

We will prove that Definition~\ref{def:pauli-flow} is equivalent to the standard one from the literature. For this, we will compare to the version that closely resembles the one given by Mitosek and Backens~\cite{mitosek2024algebraic}, which is an equivalent, more symmetric version of~\cite{browne2007gflow}.

\begin{theorem}\label{thm:pauli-flow-full}
  For an open graph $(G, I, O)$ with a choice of measurements $\mu$, a pair $(\preceq, c)$ is a Pauli flow in the sense of Definition~\ref{def:pauli-flow} if and only if:
  \begin{enumerate}
    \item[(P1)] $\mu(v) \notin \{ X, Y \}, u \npreceq v \implies v \notin c(u)$
    \item[(P2)] $\mu(v) \notin \{ Y, Z \}, u \npreceq v \implies v \notin \Odd(c(u))$
    \item[(P3)] $\mu(v) \notin \{ X, Z \}, u \npreceq v \implies v \notin c(u) \Delta \Odd(c(u))$
    \item[(P4)] $\mu(u) = XY \implies u \notin c(u) \wedge u \in \Odd(c(u))$
    \item[(P5)] $\mu(u) = XZ \implies u \in c(u) \wedge u \in \Odd(c(u))$
    \item[(P6)] $\mu(u) = YZ \implies u \in c(u) \wedge u \notin \Odd(c(u))$
    \item[(P7)] $\mu(u) = X \implies u \in \Odd(c(u))$
    \item[(P8)] $\mu(u) = Z \implies u \in c(u)$
    \item[(P9)] $\mu(u) = Y \implies u \in c(u) \Delta \Odd(c(u))$
  \end{enumerate}
\end{theorem}

\begin{proof}
  First, assume $(\preceq, \mu, c)$ satisfies conditions (i) and (ii) from Definition~\ref{def:pauli-flow}. Condition (i) implies (P4)-(P9). Condition (ii) implies (P1)-(P3).

  Conversely, assume the criteria (P1)-(P9) from the Theorem. We begin by showing condition (i) of Definition~\ref{def:pauli-flow} by case distinction on the Pauli $P \in \mu(u)$. If $X$ is in $\mu(u)$, then the predicate from either (P4), (P5), or (P7) is satisfied. In any of these cases, $u \in \Odd(c(u)) = A_X(u)$. Similarly, if $Y \in \mu(u)$, the predicate from (P4), (P6), or (P9) is satisfied. Each of these cases imply the $u$ is in the symmetric difference $c(u) \Delta \Odd(c(u)) = A_Y(u)$. If $Z \in \mu(u)$, then the predicate from (P5), (P6), or (P8) is satisfied, hence $u \in c(u) = A_Z(u)$.

  Condition (ii) follows similarly by case distinction. It can be equivalently stated as $P \in \mu(v), u \npreceq v \implies v \notin A_P(u)$, so take some node $v$ and some $u \npreceq v$. If $Z \in \mu(v)$, then (P1) implies $v \notin c(u) = A_Z(u)$. If $Y \in \mu(v)$, then (P3) implies $v \notin c(u) \Delta \Odd(c(u)) = A_Y(u)$. Finally, if $X \in \mu(v)$, then (P2) implies $v \notin \Odd(c(u)) = A_X(u)$.
\end{proof}

\section{Pauli Flow Equivalence Proofs}\label{app:pauli-flow-equiv}

\begin{proof}[Proof of Lemma~\ref{lem:strong-zxflow-iff-pauli-flow}]
Let $(\preceq, \ell_Z, \ell_X, f)$ be a strong ZX-flow for $D$. For every non-output spider $\nu$, let the correction set $c(\nu)$ be the generating set of spiders for the Pauli semiweb $f(\nu)$, as defined in Theorem~\ref{thm:generating-set-of-spiders}.

It must be the case that $f(\nu)$ has a $\pi$-defect at $\nu$. We will show that this implies the Pauli flow conditions at $\nu$ for each of the measurement labels. First note that, since $\nu$ is not adjacent to an output and inputs must be un-coloured by $f(\nu)$, any $Z$-coloured wires adjacent to $\nu$ must come from basic semiwebs in $c(\nu)$.
\begin{itemize}
\item If $X \in \mu(\nu)$, then the phase of $\nu$ is not an odd multiple of $\frac\pi2$. Hence, the only way $f(\nu)$ can have a $\pi$-defect at $\nu$ is if there is an odd number of adjacent wires with $Z$-support, which implies $\nu$ is in the odd neighbourhood of $c(\nu)$.
\item If $Y \in \mu(\nu)$, then the phase of $\nu$ is not a multiple of $\pi$. Hence, it must either be the case that (1) no adjacent wires have $X$-support and oddly many adjacent wires have $Z$-support, or (2) all adjacent edges have $X$-support and evenly-many adjacent wires have $Z$-support. This implies that $\nu$ is in the symmetric difference of $c(\nu)$ and the odd neighborhood of $c(\nu)$.
\end{itemize}
Now, consider another spider $\nu'$:
\begin{itemize}
\item If $X \in \mu(\nu')$ then the phase of $\nu'$ is not an odd multiple of $\frac\pi2$. Hence, if $\nu'$ is in the odd neighbourhood of $c(\nu)$, $\nu$ must have a $\pi$-defect at $\nu'$ so that $\nu \preceq \nu'$.
\item If $Y \in \mu(\nu')$, then the phase $\alpha$ of $\nu'$ is not a multiple of $\pi$. Hence, if $\nu'$ is in the symmetric difference of $c(\nu)$ and $\Odd(c(\nu))$, then it is either in $\Odd(c(\nu))$ so that $f(\nu)$ has a $\pi$-defect due to the odd number of $Z$-support edges, or it is in $c(\nu)$ so that $\nu'$ has all $X$-support and $f(\nu)$ has a $2\alpha$-defect there. In both cases $\nu \preceq \nu'$.
\end{itemize}
Hence, $(c,\mu,\preceq)$ satisfies the conditions of Pauli flow of Definition~\ref{def:pauli-flow}.

Now for the converse, let $(\preceq',\mu, c)$ be a Pauli flow for $D$. Let $\preceq$ be the extension of $\preceq'$ to all spiders, including output spiders, where output spiders are maximal with respect to $\preceq$.

We first define the flow semiwebs for $D$. Since we want a strong ZX-flow, we need to do it for all spiders, not just the non-Clifford ones.

For each output spider, we let $f(\nu)$ be the Pauli semiweb supported on the unique output wire adjacent to $\nu$. For each non-output spider, we let $f(\nu) := \prod_{\nu' \in c(\nu)} \vec b_{\nu'}$ where $\vec b_\nu$ is the basic semiweb at $\nu$ (Definition~\ref{def:basic-semiweb}).

If $\mu(\nu) = X$, its phase must be a multiple of $\pi$ and $\nu \in \Odd(c(\nu))$, so $f(\nu)$ has oddly many adjacent wires with $Z$-support, and hence has a $\pi$-defect. If $\mu(\nu) = Y$, its phase is $\pm\frac\pi2$ and $\nu$ is either in $c(\nu)$ or $\Odd(c(\nu))$, but not both. In either case, $f(\nu)$ has a $\pi$-defect. Finally, if $\mu(\nu) = XY$, then oddly many adjacent wires have Z-support and no adjacent wires have X-support, so it has a $\pi$-defect at $\nu$.

Now, consider another spider $\nu'$ such that $f(\nu)$ has a defect at $\nu'$. If $\nu'$ is adjacent to an output, then $\nu \preceq \nu'$. If it is a non-output, then if its phase is a multiple of $\pi$, then oddly many neighbouring wires must have Z-support in $f(\nu)$ for it to have a defect, so $\nu' \in \Odd(c(\nu))$ and hence $\nu \preceq \nu'$. If the phase of $\nu'$ is an odd multiple of $\frac\pi2$, then it can only have a defect if $\nu'\in c(\nu)$ or $\nu' \in \Odd(c(\nu))$, but not both, so that $\nu' \in A_Y(\nu)$. Hence, also $\nu \preceq \nu'$. If instead the phase of $\nu'$ is non-Clifford, $\nu' \in c(\nu)$ or $\nu' \in \Odd(c(\nu))$. So, $\nu'$ must either be in $A_X(\nu)$ or $A_Y(\nu)$, so that also $\nu \preceq \nu'$. All the flow semiwebs hence have the correct properties with respect to $\preceq$.

Finally, we construct the logical semiwebs as follows. $\ell_Z(i)$ is the edge semiweb supported on input $i$. This has $Z$-support just on input wire $i$ and no $X$-support. We construct the logical semiweb $\ell_X(i)$ as the basic semiweb $b_\nu$ of the spider $\nu$ that is adjacent to input $i$. Since the diagram is graph-like, this basic semiweb restricted to inputs is $X_i$, as required.
\end{proof}

\end{document}